\documentclass[sigconf]{acmart}
\usepackage{subcaption}
\usepackage{booktabs}
\usepackage{enumitem}
\setcopyright{none}
\copyrightyear{2018}
\renewcommand\footnotetextcopyrightpermission[1]{} 
\usepackage{ifthen}
\newcommand{\us}{\char`_}

\usepackage{graphicx}
\graphicspath{ {images/} }
\usepackage{multicol}
\usepackage{listings}
\usepackage{xcolor}
\definecolor{theWhite}{gray}{1.0}
\newcommand{\DONE}[1]{}
\newcommand{\COMMENT}[1]{}

\newcommand{\ignore}[1]{}

\newcounter{programlinenumber}

\newboolean{TR}
\setboolean{TR}{false}
\ifthenelse{\boolean{TR}}{

\newcommand{\TrOnly}[1]{#1}
\newcommand{\SubOnly}[1]{}
\newcommand{\TrOnlyInFootnote}[1]{#1}
\newcommand{\TrOnlyInTable}[1]{#1}}
{

\newcommand{\TrOnly}[1]{}
\newcommand{\SubOnly}[1]{#1}
\newcommand{\TrOnlyInFootnote}[1]{}
\newcommand{\TrOnlyInTable}[1]{}}

\usepackage{microtype}
\begin{document}
\title{Scaling Ordered Stream Processing on Shared-Memory Multicores}
\author{Guna Prasaad}
\affiliation{\institution{University of Washington}}
\email{guna@cs.washington.edu}

\author{G. Ramalingam}
\affiliation{\institution{Microsoft Research India}}
\email{grama@microsoft.com}

\author{Kaushik Rajan}
\affiliation{\institution{Microsoft Research India}}
\email{krajan@microsoft.com}

\renewcommand{\shortauthors}{G. Prasaad et al.}


\begin{abstract}
Many modern applications require real-time processing of large volumes of high-speed data. Such data processing needs can be modeled as a streaming computation. A streaming computation is specified as a dataflow graph that exposes multiple opportunities for parallelizing its execution, in the form of data, pipeline and task parallelism. On the other hand, many important applications require that processing of the stream be ordered, where inputs are processed in the same order as they arrive. There is a fundamental conflict between ordered processing and parallelizing the streaming computation. This paper focuses on the problem of effectively parallelizing ordered streaming computations on a shared-memory multicore machine. 
\par
We first address the key challenges in exploiting data parallelism in the ordered setting. We present a low-latency, non-blocking concurrent data structure to order outputs produced by concurrent workers on an operator. We also propose a new approach to parallelizing partitioned stateful operators that can handle load imbalance across partitions effectively and mostly avoid delays due to ordering. We illustrate the trade-offs and effectiveness of our concurrent data-structures on micro-benchmarks and streaming queries from the TPCx-BB~\cite{TPCxBB} benchmark. We then present an adaptive runtime that dynamically maps the exposed parallelism in the computation to that of the machine. We propose several intuitive scheduling heuristics and compare them empirically on the TPCx-BB queries. We find that for streaming computations, heuristics that exploit as much pipeline parallelism as possible perform better than those that seek to exploit data parallelism.     
\end{abstract}

\maketitle
\section{Introduction}
\label{sec:introduction}
Stream processing as a computational model has a long history dating back to Petri Nets in the 1960s, Kahn Process Networks and Communicating Sequential Processes in the 1970s, and Synchronous Dataflow in the 1980s \cite{Thies2009}. 
Practical applications of stream processing were, for a long time, limited to audio, video and digital signal processing that typically involves deterministic, high-performance computations. 
Several languages and compilers such as StreamIt~\cite{StreamIt}, Continuous Query Language(CQL)~\cite{CQL} and Imagine~\cite{Imagine} were designed to specify and optimize the execution of such programs on single and shared-memory architectures.
\par
The emergence of sensors and similar small-scale computing devices that continuously produce large volumes of data led to the rise of many new applications such as surveillance, fraud detection, environment monitoring, etc.
The scale and distributed nature of the problem spurred the interest of several research communities that further gave rise to large scale distributed stream processing systems such as Aurora \cite{Aurora}, Borealis \cite{Borealis}, STREAM \cite{STREAM} and TelegraphCQ \cite{TelegraphCQ}. 
\par
In today's highly connected world, data is of utmost value as it arrives. 
The advent of Big Data has further increased the importance of realtime stream processing. 
Several modern use-cases like shopping cart abandonment analysis, ad serving, brand monitoring on social media, and leader board maintenance in online games require realtime processing of large volumes of high-speed data. 
In the past decade, there has been a tremendous increase in the number of products (e.g. IBM Streams~\cite{IBMStreams}, Millwheel~\cite{MillWheel}, Spark Streaming~\cite{SparkStreaming}, Apache Storm~\cite{Storm}, S4~\cite{S4}, Samza~\cite{Samza}, Heron~\cite{Heron}, Microsoft Stream Insight~\cite{StreamInsight}, TRILL~\cite{Trill}) that cater to such data processing needs and is evidence of its ever-growing importance.
\par
In this paper, we focus on scaling stream processing on the shared-memory multicore architecture. 
We consider this an important problem for several reasons.
Streaming pipelines generally have a low memory footprint as most of the operators are either stateless or have a small bounded state.
With increasing main memory sizes and prevalence of multi-core architectures, the bandwidth and parallelism offered by a single machine today is often sufficient to deploy pipelines with large number of operators~\cite{IBMStreams}.
So, most streaming workloads can be efficiently handled in a single multicore machine without having to distribute it over a cluster of machines.
In fact, systems like TRILL~\cite{Trill} run streaming computations entirely on a single multicore machine at scales sufficient for several important applications.  
This is unlike batch processing systems, where the input, intermediate results and output data are often large and hence the pipeline needs to be split into many stages. 
\par 
Even in workloads where distribution across a cluster is important (e.g. for fault tolerance), typically the individual nodes in the cluster are shared-memory multicores themselves. 
Most distributed streaming systems \cite{SparkStreaming, Storm} today assume each core in a multicore node as an individual executor and fail to exploit the advantages of low overhead communication offered by shared-memory. 
We believe that considering a multicore machine as a single powerful node rather than as a set of independent nodes can help better exploit shared-memory parallelism. 
A streaming computation can be split into multiple stages and each stage can be deployed on a shared-memory node~\cite{IBMStreams}.
Most prior work in this area have not studied the shared-memory multicore setting in depth - they either focus on the single core \cite{AuroraScheduling, Chain2003, Jiang2004} or distributed shared-nothing architectures \cite{Borealis, Heron, SparkStreaming, MillWheel, Storm, Dataflow}.
\par
Further, we are interested in ordered stream processing. 
The stream of events/tuples usually have an associated notion of \emph{temporal ordering} such as in user click-streams from an online shopping session or periodic sensor readings. In many scenarios the application logic depends on this temporal order. For example, clustering click-streams into sessions based on timeouts between two consecutive events and computing time-windowed aggregates over streams of data. Implementing such logic on systems that do not provide ordered semantics is complicated and often turns out to be a performance bottleneck.
\par
Ordered processing also enables our parallelization framework to be deployed easily on individual multicore nodes in a distributed stream processing cluster. This guarantee is important especially when a large stream processing query is divided into sub-queries, each allotted to a multicore node and one of them contains a non-commutative operation. Moreover, fault tolerance techniques such as active replication depend on state of the pipeline on two replicas being the same and it cannot be guaranteed without ordered processing. 
\par
Many stream processing systems today provide mechanisms to support ordered stream processing. 
Most of them based on a micro-batch architecture~\cite{Trill, SparkStreaming, MillWheel, Borealis}, in which the input stream is broken down into streams of smaller batches and each batch is processed like in a batch processing system such as Map Reduce~\cite{MapReduce} or Apache Spark~\cite{Spark}. 
They support order sensitive pipelines by periodically sending \emph{watermarks} denoting that all events less than a specific timestamp have been received. 
However, These techniques are not suitable for latency critical applications mainly due to the batching delays. 
We show that it is possible to achieve this guarantee at much lower latencies without constraining execution of the pipeline excessively.

\subsection{Background and Challenges}
A \emph{streaming computation} can be specified as a dataflow graph, where each vertex is associated with an operator and directed edges represent flow of input into and out of the operators. At runtime, every vertex receives a stream of values (which we refer to as \emph{tuples} henceforth) along each of its incoming edges. These tuples are then processed by the operator to produce zero or more output tuples that are then sent along its outgoing edges. 
\par
A unary operator processes an input (of some type $T_{in}$) and produces a sequence of zero or more outputs (of some type $T_{out}$). Every vertex with a single incoming edge has an associated unary operator that specifies the computation to be performed at that vertex. \texttt{map}, \texttt{filter} and \texttt{windowed-aggregate} are examples of such unary operators. A vertex with $n$ edges abstractly represents an $n-$ary operator, with $n$ inputs of types $T_1, T_2, ..., T_n$.
In the streaming setting, the semantics of an $n$-ary operator too can be specified as a function that maps a tuple on a (specified) incoming edge to a sequence of zero or more output tuples.
\par 
Some operators are pure functions that do not have any state associated with its computation and hence called \emph{stateless operators}. Some operators have an internal state that is accessed and updated during the computation  - for example, \texttt{windowed-count} maintains the count of tuples in the current window as internal state. Such operators are called \emph{stateful operators} . In some cases, the operator accesses only a part of the state during the computation, which is pre-determined by a \emph{key} associated with every input tuple. These are called \emph{partitioned stateful operators}, as the state can be partitioned by the key. \texttt{windowed-group-by-count} is an operator of this type.

\subsubsection{Opportunities for Parallelization}
\begin{figure*}
	\centering
	\includegraphics[scale=0.34]{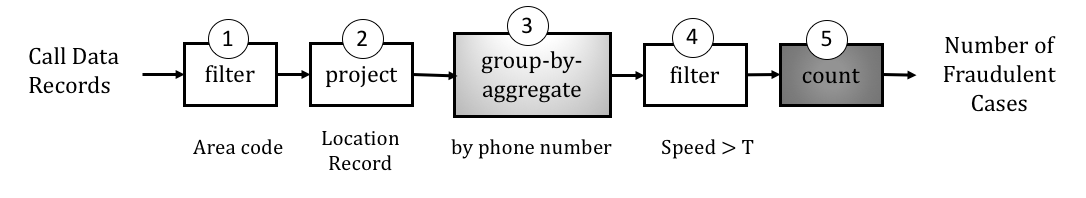}
	\caption{Algorithm for high-mobility fraud detection expressed as a streaming dataflow graph.}
	\label{fig:high-mobility-fraud-diagram}
\end{figure*}
A streaming dataflow graph exposes various opportunities for parallelizing the computation efficiently. We elucidate this using an example: figure \ref{fig:high-mobility-fraud-diagram} represents an algorithm to detect high-mobility fraud using call data records as a streaming dataflow graph.
\par
Call data records (CDR) are generated by every call between two mobile phones and it contains information such as time, duration of the call, location and phone number of the caller and the callee. In the detection algorithm, a CDR is first filtered (1, fig. \ref{fig:high-mobility-fraud-diagram}) on the interested area code and the caller/callee's time and location information is projected (2) as a record. These location records are then grouped by phone number to compute (3) the speed at which a user must have traveled between locations. Phone numbers that have a speed greater than $T$, are then filtered (4), and the number of such cases in a given time window are counted (5).
\par
An operator is said to be \emph{data parallel}, if its inputs can be processed concurrently. Stateless operators such as (1, 2, 4) in the example are data parallel. On the other hand, inputs to a partitioned stateful operator can be processed in parallel only if they belong to different partitions. Hence, they are said to exhibit \emph{partitioned parallelism}. In our example, computing the speed based on location records (3) for two different phone numbers can be done in parallel. Non-commutative stateful operators do not exhibit any data parallelism.  
\par
Further, when two operators are connected to each other such that the output of one forms the input to another, they are said to exhibit \emph{pipeline parallelism}. In that case, these two operators can be processed concurrently. For example, one worker can compute the speed (3) for a particular phone number, while another filters (4) some phone numbers based on the speed already computed and sent to be filtered. Finally, a dataflow graph also exposes \emph{task parallelism}. If two operators are not connected to each other via an input-output relationship, directly or indirectly, they can be processed concurrently. For example, operators on two sibling nodes in a DAG exhibit this kind of parallelism.

\subsubsection{Ordered Processing}
Ordered processing specifies that processing of inputs to a streaming pipeline must be semantically equivalent to executing them serially one at a time in the order of their arrival. We achieve this by ensuring that each individual operator implementation guarantees ordered processing and hence by extension any pipeline built by composing these implementations provide the ordering guarantee. 
\par
There is a fundamental conflict between data parallelism and ordered processing. Data parallelism seeks to improve the throughput of an operator by letting more than one worker operate on the inputs from the worklist concurrently. On the other hand, ordered processing requires to process them in the order of their arrival. The key observation here is that depending on the type of the operator, a concurrent execution might still be semantically equivalent to a serial single-threaded execution. 
\par
Ordered processing for a stateful operator is straightforward as its maximum allowed degree of parallelism is 1. In case of stateless and partitioned stateful operators, however, multiple workers process inputs concurrently and so they need special constructs to ensure that their concurrent execution is equivalent to a serial single-threaded execution. 
\par
There are essentially two kinds of ordering requirements that must be handled correctly. The first kind is \emph{processing order}: for some operators we need to ensure that the processing logic of the operator is executed on the inputs in the same order as they arrive. This is a key requirement for non-commutative stateful operators. On the other hand, there is no such constraint for stateless operators. Partitioned stateful operators present an interesting middle-ground where it is enough to guarantee that tuples with the same key are processed in their arrival order. 
\par 
The second kind of requirement is \emph{output ordering}, which specifies that the outputs of an operator are sent to the downstream operator in the same order as its inputs. In particular, even when inputs $i_1$ and $i_2$ can be processed concurrently (when they belong to different partitions or when the operator is stateless), we still need to ensure that the outputs $o_1$ and $o_2$ produced by these inputs respectively are sent out in the right order. We guarantee this property for both stateless and partitioned stateful operators using special concurrent data structures. We describe a low-overhead, non-blocking solution to this problem in Sec. \ref{sec:reordering-scheme}.
\par
Output ordering is innately a blocking constraint: even if $o_2$ is produced before $o_1$, it gets blocked until $o_1$ is produced and sent downstream. This manifests as an implicit advantage for parallelization schemes that processes inputs from the worklist \emph{almost} in the order of their arrival even though the semantics does not impose a restriction on this order. For stateless operators, having a shared worklist directly enables this execution pattern. However, it is non-trivial to achieve this for partitioned stateful operators. We present an adaptive partitioning scheme that supports this notion of almost ordered processing in Sec. \ref{sec:partitioned-parallelism}.
 
\subsection{Contributions}
In this paper, we make the following contributions:  
\begin{enumerate}[leftmargin=\parindent]
\item (Sec. \ref{sec:reordering-scheme}) We present a low-overhead non-blocking reordering scheme to order outputs of an operator that are produced concurrently. 
We observe that it scales better than a standard lock-based scheme and overall provides better throughput for long pipeline queries.
\item (Sec. \ref{sec:partitioned-parallelism}) We propose a novel scheme for exploiting partitioned parallelism in the ordered setting. 
We observe that our scheme achieves better speedup than the predominantly used strategy for partitioned parallelism during partition-induced skews and mostly avoids delay due to ordering constraints leading to much lower latencies. 
\item (Sec. \ref{sec:dynamic-scheduling}) We propose several intuitive scheduling heuristics that can be used to dynamically schedule operators at runtime. We identify a single heuristic that produces the best throughput and near best latency.  
\item (Sec. \ref{sec:evaluation}) We evaluate our runtime on streaming queries from TPCx-BB\cite{TPCxBB} and demonstrate that we can provide a throughput of millions of tuples per second on some queries with latency in the order of few milliseconds. 
\end{enumerate}

\section{Solution Overview}
\label{sec:overview}
A generalized solution model for executing a stream processing query comprises of two components: a \emph{compiler} and a \emph{runtime}. The compiler is responsible for static optimizations, while the runtime takes this compiled representation and executes it on the machine, potentially with dynamic optimizations. 
\par
The relative roles of the compiler and runtime are determined by the type of streaming computation. For example, the streaming computations in signal processing are deterministic, and operator characteristics (such as per-tuple processing cost, selectivity) are known a priori. Such workloads provide more opportunities for static compiler optimizations, and the runtime is a straightforward execution of the produced scheduling plan. This model is exemplified by systems like StreamIt \cite{StreamIt} and Brook\cite{Brook}. In other applications like monitoring, fraud detection or shopping cart analysis there is little to no information about the operator characteristics at compile-time and hence the scope of static optimizations are fewer. So, systems like Borealis\cite{Borealis} and STREAM \cite{STREAM} designed for these workloads rely heavily on dynamic optimizations. However, even in such dynamic workloads there is some scope for static optimizations like coarsening of operators, pushing up filters. Refer \cite{Survey2014} for a detailed catalog of such optimizations.
\par
In our system, we target dynamic workloads to support use-cases that have risen in many new Big Data applications. We assume that a stream processing query is initially compiled into an optimal pipeline using some of the known techniques. We then deploy this optimal version of the pipeline on a runtime that seeks to efficiently parallelize its execution with the ordered processing guarantee. Here we focus only on the design of runtime, as the compilation stage is quite well studied in earlier works. We limit our discussion to linear chain pipelines, which is the predominant structure present in most stream processing queries. We believe our ideas can be generalized to other DAG structures as well, but we do not specifically address them here. 

\subsection{Problem Definition} 
The system accepts a pipeline that consists of operators connected to each other as a linear chain. The operators are specified to be either of stateless, stateful or partitioned stateful type. In case of partitioned stateful operators, the user also specifies a key selector that can be used to associate tuples with keys and a partitioning strategy such as hash or range partitioning to further map keys to partitions. The goal of the runtime is to execute this linear pipeline efficiently on a shared-memory multicore machine by exploiting various forms of parallelism as described in Sec. \ref{sec:introduction}. 
\par
There are two dimensions of performance for a stream processing system that we are interested in. First is the throughput, by which we refer to the number of tuples processed to completion every second. The second dimension is latency. There are two notions of latency prevalent in the literature: end-to-end latency, which is the time duration between entry at ingress and exit at egress, and processing latency, which is the time since the first operator begins processing a tuple until exit at egress. The difference between them is that end-to-end latency includes the time spent by the tuple in the input queue for the overall pipeline. In the rest of this paper, we refer to processing latency when we say latency. The objective here is to maximize the throughput to handle high-speed data while minimizing the latency to process them in realtime.

\subsection{Runtime Design}  
Our runtime is based on an \emph{asynchronous} model of execution. We first decouple the pipeline into individual operators and compile them to independently schedulable units, one for every operator. We do this by associating every operator with a worklist(s). Inputs to an operator are simply added to its worklist instead of executing the operator logic synchronously.  When the operator is scheduled, it obtains inputs from the worklist, processes them and adds the outputs to the worklist of the downstream operator. 
\par
The goal of the runtime is to choose which operator to choose, at what time and on which core? The two essential components of our runtime are \emph{worker threads} and the \emph{scheduler} data structure. The worker threads are the work horses of our runtime and responsible for advancing the progress of operators. Worker threads periodically query the scheduler for work. A worker, when allotted to an operator, dequeues an input from the operator's worklist, performs the operation and adds the output(s) produced to worklist of the next operator in the pipeline. A worker is specified with the maximum number of tuples to process in an operator and when allotted it processes as many tuples before deciding which operator to work on next. The worker additionally collects runtime information about each operator such as number of inputs consumed, outputs produced, time taken to process them. This is then used to estimate operator characteristics like average per-tuple processing cost and average selectivity. 
\par
Scheduling decisions regarding which operator must be scheduled next are made by a central scheduler data structure. This decision is made using estimated operator characteristics, current worklist sizes, and possibly observed throughput and latency measurements. We achieve this using scheduling heuristics - we discuss several of them in Sec. \ref{sec:dynamic-scheduling}. When a heuristic chooses to schedule two different operators on different cores, it seeks to exploit pipeline parallelism. When it schedules the same operator on different cores it exploits data parallelism ingrained in the operator. Overall the goal of the scheduler is use to \emph{dynamically} determine an ideal combination of data and pipeline parallelism among operators to achieve optimal performance.
\par
The scheduler in our runtime can dynamically schedule more than one worker on an operator. This is applicable only to data or partition parallel operators as the maximum degree of parallelism allowed by a stateful operator is 1. The implementation of these operators internally handle the required concurrency control to ensure correct and ordered processing (refer Sec. \ref{sec:correctness}). This is unlike many other architectures~\cite{IBMStreams}, where a single logical operator is replicated into a statically determined number of physical operators that are then scheduled independently. 
\section{Reordering Scheme}
\label{sec:reordering-scheme}
In this section, we handle the problem of ordering outputs produced by concurrent workers before they are sent to downstream operator(s). Most prior solutions to this problem are restricted to the micro-batching architecture: the input tuple stream is considered as a stream of batches, where tuples in a batch are executed in parallel and their outputs are finally sorted before sending them downstream. The notion of batching has some advantages including amortizing the cost involved in sorting, admitting columnar-based and operator-specific batch optimizations. However, these solutions are predominantly known to trade off latency for throughput. Our approach seeks to perform this reordering incrementally using low overhead non-blocking concurrent data structures.
\par
For stateless and partitioned stateful operators, multiple workers can consume inputs from their worklist producing outputs concurrently. Each input is associated with a unique serial number (starting from 1) denoting its arrival order into the worklist of the operator. This serial number is assigned using an atomic counter at the time of enqueueing them to the worklist(s) of the operator. In some cases, a single input can produce more than one outputs. However, they are considered together as one unit and is associated with a single serial number. The schemes we describe below are concerned only with ordering outputs based on this serial number.
\par
Specifically, the ordering constraint requires that for all $t$, the output $o_t$ produced by a tuple $i_t$ be sent downstream (either to an operator or egress) only after $o_1, o_2, ..., o_{t-1}$ are sent downstream. Since, these outputs are produced by concurrent workers, they are produced in no predetermined order. So, $o_{t+1}$ might be produced before $o_{t}$ and in that case $o_{t+1}$ has to wait until $o_{t}$ is produced and sent. We first describe a lock-based solution that implements this waiting scheme. We show that such a straight-forward design could lead to sub-optimal performance. Then, we present our improved low-latency, non-blocking solution.

\subsection{Lock-Based Solution.}
A standard approach would be to use a waiting buffer and a counter. The counter keeps track of the serial number of next output to be sent. Whenever the corresponding output is available it is sent downstream immediately and the counter incremented. If an output is not the next one to be sent, we simply add it to the waiting buffer and return to process more inputs. So, when an output is sent, we must check the waiting buffer for the next output and if present, send that to the downstream operator and repeat. Further, we do not want multiple workers to send the output(s) downstream or increment the counter concurrently as that will violate our ordering guarantee. So, we protect the overall logic using a global lock to ensure correctness and progress. This scheme is listed in fig. \ref{alg:naive-reordering-scheme}
\begin{figure}[t]
\centering
\begin{minipage}[b]{0.4\textwidth}
\begin{lstlisting}[mathescape, numbers = left]
void send($o_t$) {
 lock();
 if ($t$ == next) {
  send_downstream($o_t$);
  next++;
  while(buffer has $o_{next}$) {
   send_downstream($o_{next}$);
   next++;
 } } else {	
   add $o_t$ to buffer
 }
 unlock();
}
\end{lstlisting}
\caption{Lock-based Scheme: The global lock used here induces unnecessary blocking behavior}
\label{alg:naive-reordering-scheme}
\end{minipage}
\hspace{0.5cm}
\begin{minipage}[b]{0.45\textwidth}
\centering
\includegraphics[scale=0.6]{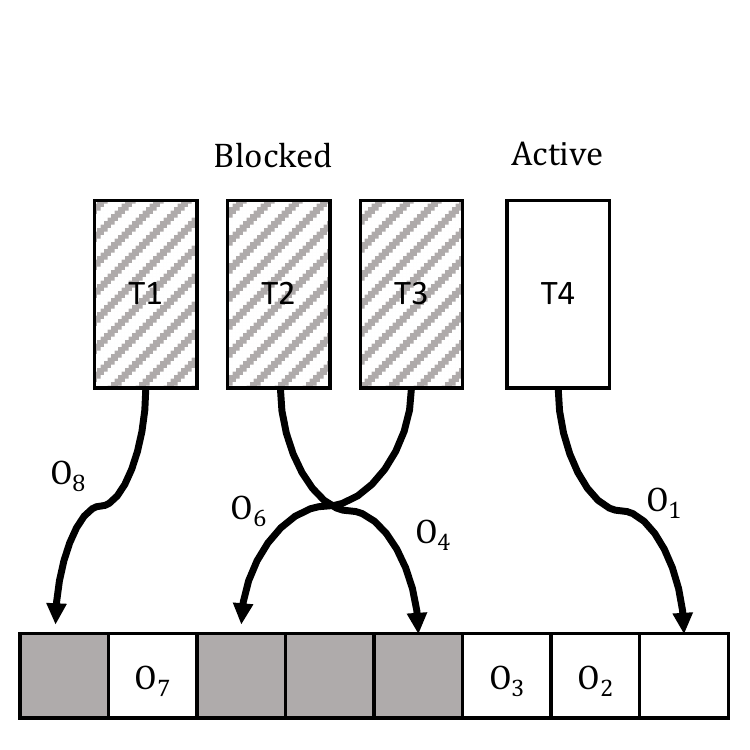}
\caption{Unnecessary Blocking: $T_1, T_2$ and $T_3$ are blocked until $T_4$ sends outputs $O_1, O_2$ and $O_3$ downstream.}
\label{fig:naive-problem}
\end{minipage}
\end{figure}

However, this scheme results in sub-optimal performance due to unnecessary blocking of workers. Consider the scenario shown in fig. \ref{fig:naive-problem}: worker thread $T_4$ produces $O_1$, which is the next output to send downstream, obtains the lock and keeps sending outputs $O_2$ and $O_3$ as they are already available in the waiting buffer. Meanwhile, workers $T_1, T_2, T_3$ that produced outputs $O_4$, $O_6$and $O_8$ respectively get blocked trying to acquire the lock. However, we know that outputs have pre-allotted serial numbers. So, adding them to the waiting buffer can be totally independent of sending them downstream. Ideally, workers $T_1, T_2$ must be able to add their outputs to the waiting buffer while another worker is sending outputs downstream and return back to do useful work. 

\subsection{Non-Blocking Solution}
We improve this version by replacing the \texttt{lock} with an atomic flag, essentially to provide \texttt{try\us lock} semantics. This scheme is listed elaborately in fig. \ref{alg:non-blocking-reordering-scheme}. Any worker $w$ seeking to send an output downstream, first tries to add it in a bounded circular buffer. The buffer is used to store available outputs that are not yet ready to be sent. This step can either fail or succeed based on the size of buffer and current value of the \texttt{next} counter. If it fails, the worker tries again with the same output, after it exits the \texttt{send} function. 
\par
Before exiting, irrespective of success or failure in the add step, $w$ tries to send pending outputs in the buffer to downstream operator(s). It can do so only when it can \texttt{test\us and\us set} a global atomic flag. If it cannot set the flag, it means that another worker $w'$ is performing this step.  In that case, $w$ simply exits the function instead of getting blocked, unlike in the lock-based scheme.
\par
If $w$ can set the flag, it has exclusive access to send the buffered outputs. First, it obtains the current value of \texttt{next} counter and the corresponding value from the \texttt{buffer} array. If this value is not \texttt{EMPTY}, it sends the output downstream, increments the counter and repeats this again for the new value of \texttt{next}. 
If the obtained value is \texttt{EMPTY} then, $w$ clears the flag and exits the loop. Further, to ensure that every output is sent downstream as soon as it is ready to be, $w$ checks the buffer array again and retries to send the previously unavailable output, if it is available now. This ensures that there is no ready-to-send output in the buffer, when there are no active workers inside \texttt{send}.   
\par
\begin{figure}
	\centering
	\begin{lstlisting}[mathescape, numbers = left]
	//data fields
	atomic_long next;
	atomic<output*> buffer[$s$];
	atomic_flag flag;

	//invoked by workers
	bool send($o_t$) {
		bool success = try_add($o_t$);
		send_pending_outputs();
		return success;
	}
	
	//helper functions
	bool try_add($o_t$) {
		$n$ = next.load();
		if($t \geq n$ and $t < n + s$) {
			$i = t \mod s$;
			buffer[$i$].set($o_t$);
			return true;
		} else {
			return false;
	} }

	void send_pending_outputs() {
		if (not flag.test_and_set()) {
			//send as many outputs as possible
			while(true) {
				$n =$ next.load();
				$i = n \mod s$;
				$o$ = buffer[$i$].load();
				if ($o$ is not EMPTY) {
					send_downstream($o$);
					buffer[$i$].set(EMPTY);
					next.fetch_add(1);
				} else {
					flag.clear();
					break;
			}	}
			//re-check if next output is available
			$o$ = buffer[$i$].load();
			if ($o$ is not EMPTY) {
				send_pending_outputs();
	}	}	}
	\end{lstlisting}
	\caption{Non-blocking Reordering Scheme}
	\label{alg:non-blocking-reordering-scheme}
\end{figure}

\begin{theorem}[Correctness of Non-Blocking Reordering Scheme]
	If all concurrent workers allotted to an operator send outputs to operators downstream by invoking the \texttt{send} procedure (fig. \ref{alg:non-blocking-reordering-scheme}), then output $o_t$ (with serial number $t$) is sent downstream (by invocation of \texttt{send\us downstream}) only after all outputs $o_1, o_2, ..., o_{t-1}$ are sent. 
\end{theorem}
\begin{proof}
	The outputs are sent downstream only inside the send-pending-outputs procedure, in which lines L27-36 (referred to as \emph{exit section}) are protected from concurrent access by the atomic flag variable \texttt{flag}. Since this makes the exit section a critical section, at most one worker increments the \texttt{next} counter and sends the pending outputs in the buffer to the operator downstream. It is quite clear from the control flow in the exit section, that whatever non-\texttt{EMPTY} output is present in \texttt{buffer[$i$]}, it is sent downstream as the output with serial number $n$, where $i = n \mod s$.  Now, it suffices to prove that if the value of \texttt{next} is $n$ and $o$ is the value obtained by loading \texttt{buffer[$i$]} as in fig. \ref{alg:non-blocking-reordering-scheme}, then the following two conditions hold: 
	\begin{enumerate}[leftmargin=\parindent]
		\item If $o_n$ has not been added to the buffer, then $o$ is EMPTY
		\item If $o$ is not EMPTY, then the value of $o$ is $o_n$
	\end{enumerate}
	In order to prove this, we first define $T_k$ to be the time at which the \texttt{next} counter is atomically incremented from $k$ to $k+1$. For simplicity of explanation, we assume $T_k$, for $k < 0$, to be some global initialization time when \texttt{buffer} array is initialized with \texttt{EMPTY}. 
	\par
	The condition at L16 (referred as \emph{entry condition}) determines whether an output $o_t$ (with serial number $t$) can be added to the \texttt{buffer} at $i = (t \mod s)$ or not. This condition enforces that $o_t$ can be added only when \texttt{next} $\in (t - s, t]$, which in turn can happen only during the time interval $(T_{t-s}, T_{t})$.  
	\par
	 Since all updates to the global data fields are atomic, they are sequentially consistent. So, the value of \texttt{buffer[$i$]} (where $i = t \mod s$) is set to \texttt{EMPTY} in the exit section before $T_{t-s}$. We also know that this will definitely remain \texttt{EMPTY} until $T_{t-s}$. This is because the only valid output that can be added at $i$ during that intermittent time is $o_{t-s}$ due to the entry condition. But, we know $o_{t-s}$ has already been added once and by uniqueness of serial numbers, we can assert it will not be added again. From just after $T_{t-s}$, this value will still remain \texttt{EMPTY}  until some worker adds an output into that slot. Again, the entry condition now ensures that only $o_t$ can be added to \texttt{buffer[$i$]} during $(T_{t-s}, T_{t})$. Hence, (1) holds. 
	 \par
	  Further, control flow in the exit section necessitates that $o_t$, if available, is read into $o$ before $T_{t}$. This together with the entry condition, ensures that $o_t$ is not overwritten before being read back from \texttt{buffer[$i$]}. Hence, (2) holds. Both conditions (1) and (2), in addition with the guarantee that \texttt{next} cannot have a value $k+1$ before $k$, we can assert that the outputs are indeed sent downstream in the serial order.   
\end{proof}

\paragraph{Progress} In the above scheme, none of the concurrent workers get blocked due to another worker sending outputs. However, a worker can get blocked due to limited size of the waiting buffer: when it tries to send an output that corresponds to input with a serial number much higher than the current value of \texttt{next}, it can potentially get blocked trying and failing repeatedly to add the output. This is because the entry condition prevents this output to be added until some earlier outputs are sent and the buffer makes space for this output. Meanwhile, this worker repeatedly tries to send it and fails. 
\par 
One simple way to handle this would be to use a non-blocking concurrent map instead of a bounded array. However, the overheads in a simple array are much lesser compared to the alternatives and hence we chose such a design. Even though, we can never eliminate this scenario with a bounded buffer, we can try to avoid its occurrence as much as possible.  One could use an appropriately sized waiting buffer. Further, we could employ design strategies such that concurrent workers working on a data parallel operator would produce outputs almost in-order of their serial numbers. This is a key design strategy in exploiting partitioned parallelism in the ordered setting, which we present in the next section.
\section{Partitioned Parallelism}
\label{sec:partitioned-parallelism}
The essence of partitioned parallelism is that every input to be processed has a \emph{key}, and the state required to process inputs with different keys are disjoint. 
This allows us to process tuples with different keys in parallel, though those with the same key must be processed sequentially, in order. 
\par 
The key space can be statically partitioned into many disjoint buckets based on a strategy such as range or hash partitioning. The system treats tuples belonging to the same bucket as potentially having the same key and processes them sequentially. If the number of buckets is $p$, it limits the degree of parallelism to $p$. Ideally, we would like to have as many buckets as the number of keys to exploit as much parallelism as possible even during load imbalance induced by the partitioning strategy. But, scheduling overheads and complexities in the key space force us to have a fewer, fixed number of buckets. However, a more fine-grained partitioning strategy is still preferable, given the overheads are admissible. Profiling data gathered from sample runs can be used to determine both $p$ and the partitioning of key space into $p$ buckets. 
\par
Further, we would like to design a flexible scheme where workers can be dynamically allotted to operators. This is necessary to support a dynamic scheduling based runtime that allots workers to operators based on current status of the pipeline. As we saw in Sec. \ref{sec:reordering-scheme}, we would also like the processing order of inputs belonging to different buckets to be as close to arrival order as possible. This is because reordering of outputs will lead to unnecessary blocking if processed too much out-of-order. In the rest of this section, we describe the concurrent data-structure and strategy we employ to achieve ordered partition parallelism. We first describe two simpler strategies for implementing such an operator before presenting our approach.

\subsection{Shared-Queue Approach.} In the first approach, which we refer to as the \emph{shared-queue} approach, the producers (preceding operators in the dataflow graph) enqueue their outputs to a single queue (the worklist), and all concurrent workers extract tuples from the same queue and process them. This is a fairly straightforward strategy when the operator is stateless. We can use any linearizable concurrent queue to support multiple producers and consumers. 
\par
\begin{figure}
	\begin{minipage}[b]{0.46\textwidth}
		\centering
		\includegraphics[scale=0.15]{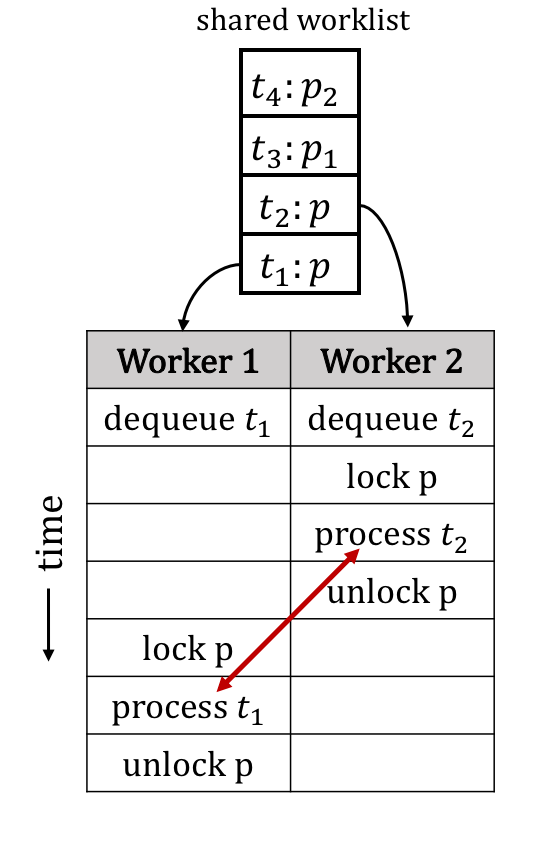}
		\caption{Shared-Queue Approach: Each worker must dequeue the tuple and obtain a lock on the tuple's bucket atomically, otherwise a concurrent execution might violate the processing order constraint as shown above.}
		\label{fig:shared-queue-diagram}
	\end{minipage}
	\hspace{0.5cm}
	\begin{minipage}[b]{0.46\textwidth}
		\centering
		\includegraphics[scale=0.15]{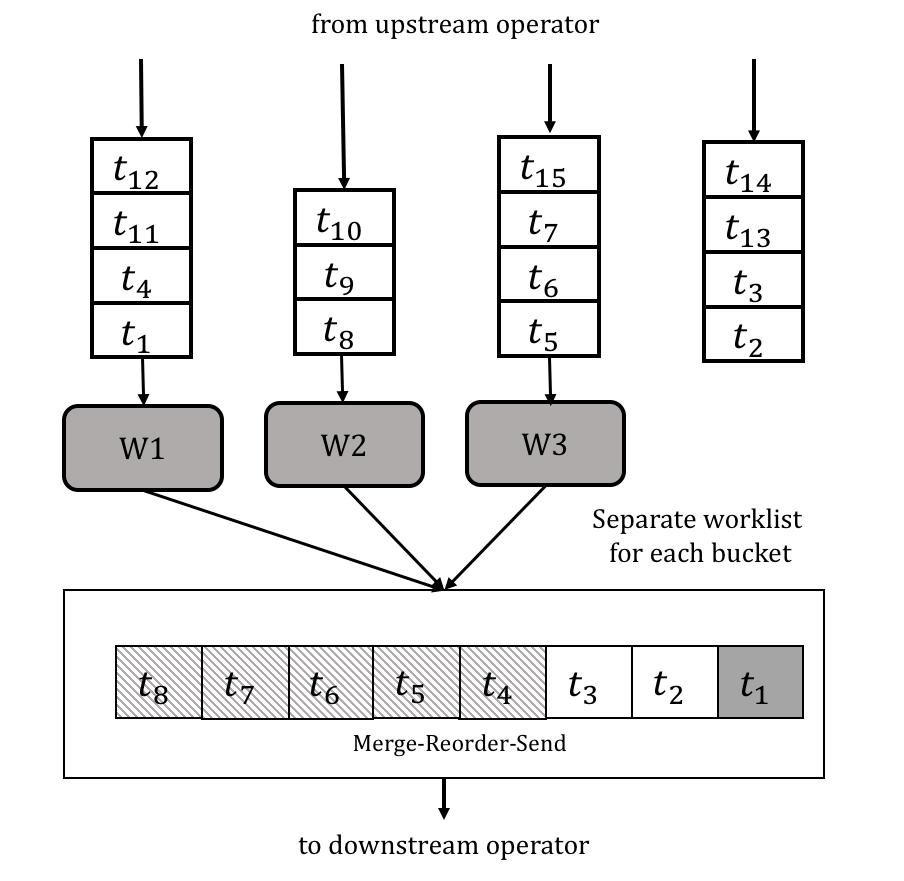}
		\caption{Partitioned-Queue Approach: Since no worker is allotted to the last bucket which contains tuple $t_2, t_3$, the outputs of tuples $t_4$ to $t_{12}$ will get blocked from flowing to the downstream operator, limiting pipelined parallelism.}
		\label{fig:partitioned-queue-diagram}
	\end{minipage}
\end{figure}
A partitioned operator, however, introduces a key challenge: we need to ensure that the items with same key are processed sequentially and in order. A naive approach would be as follows: Each worker first dequeues an item $t$ and then acquires a lock (or use any equivalent mechanism to ensure isolation) on the item's key so that two items with the same key are not processed concurrently. However, these two actions must be performed \emph{atomically}: otherwise, two workers could concurrently dequeue items $t_1$ and $t_2$ with the same key $k$, but end up acquiring the lock on $k$ out-of-order and thus process them out-of-order as shown in fig. \ref{fig:shared-queue-diagram}. This necessitates quite complex and expensive concurrency control. Furthermore, this also introduces potentially blocking behavior when one worker waits for another, which is processing an input tuple with the same key. A naive implementation could aggravate this, causing all workers to be blocked, if a global lock is used to ensure the atomicity of the sequence of these two actions.

\subsection{Partitioned-Queue Approach.} 
The second approach, which we refer to as the \emph{partitioned-queue} approach avoids this problem. We use separate queues (worklist) for each bucket and the producers enqueue each tuple into the queue corresponding to the tuple's key. Different workers process different queues and hence there is no need for explicit concurrency control. However, this approach has its own set of drawbacks: Consider the scenario shown in figure \ref{fig:partitioned-queue-diagram}, where the number of workers assigned to an operator is less than $p$ (number of buckets). In this case, the workers may make progress processing a subset of the $p$ queues. However, the outputs produced by these workers will be blocked by the reordering scheme that merges the outputs produced from the $p$ queues in order. This can cause further sub-optimal performance downstream as this behavior limits available pipelined parallelism between this and the downstream operator.

\subsection{Hybrid-Queue Approach} 
We propose a hybrid approach that combines techniques from both these strategies. We use separate queues, one for each bucket as described above. In addition, we utilize a master queue which is analogous to the single queue of the former approach. Actual tuples are stored in individual bucket queues while, the master queue stores the \emph{key} of each tuple. We list the execution model in fig. \ref{alg:hybrid-queue}.
\par
Every worker $w$ dequeues a key $k$ from the master queue, and then tries to gain exclusive access to the queue $Q_k$ that corresponds to $k$. 
If some other worker $w'$ already has exclusive access to queue $Q_k$, then worker $w$ delegates the responsibility of processing the corresponding tuple to $w'$, by incrementing a concurrent counter $count_k$  associated with the key $k$. 
The counter $count_k$ denotes the number of tuples from $Q_k$ to be processed before the active worker of key $k$ ($w'$ in this case) tries to dequeue the next key from master queue. The same counter is used to provide exclusive access to the queue $Q_k$. Having delegated the responsibility of processing the dequeued tuple to $w'$, worker $w$ can return to process the next key from the master queue.
\par
If, on the other hand, worker $w$ gains exclusive access to queue $Q_k$, it dequeues the next tuple from $Q_k$ and processes it. 
However, after processing it, the worker needs to check if there are any delegated tuples that it needs to process from the same queue $Q_k$. 
As long as the concurrent counter $count_k$ indicates there are delegated items, the worker continues to dequeue tuples from $Q_k$ and processes them.
When the counter becomes zero, the worker returns to processing the master queue. 
We prove the correctness of this scheme in the theorem below.

\begin{figure}
\centering
\begin{lstlisting}[mathescape, numbers = left]
//invoked by producers
void addInput(tuple) {
	$p$ = getPartition(tuple);
	partitionQueues[$p$].enqueue(msg);
	masterQueue.enqueue($p$);
}
//invoked by workers
void consumeInputs() {
	while(masterQueue.tryDequeue($p$)) {
		if(count[$p$].fetch_add(1) == 0) {
			do {
				partitionQueues[$p$].tryDequeue(tuple);
				operate(tuple);
			} while(count[$p$].fetch_sub(1) > 1) ;
} } }
\end{lstlisting}
\caption{Hybrid Queue Approach: The \texttt{addInput} procedure is invoked by upstream operators and \texttt{consumeInputs} procedure is invoked by workers allotted to the partitioned stateful operator}
\label{alg:hybrid-queue}
\end{figure}

\begin{theorem}[Correctness of hybrid-queue algorithm]
	If inputs to a partitioned stateful operator $\textbf{o}$ are added using the \texttt{addInput} (fig. \ref{alg:hybrid-queue}) procedure and workers allotted to $\textbf{o}$, consume inputs by invoking the \texttt{consumeInputs} procedure (fig. \ref{alg:hybrid-queue}), then the following properties hold:
	\begin{enumerate}
		\item No two workers can \texttt{operate} on tuples having the same key $k$ concurrently
		\item All tuples that have the same key $k$ are processed exactly once and in the order of their arrival
	\end{enumerate}
\end{theorem}

\begin{proof}
	Any worker allotted to the operator first dequeues a partition $p$ from the master queue. The condition at L10 in fig. \ref{alg:hybrid-queue} ensures that a worker can obtain a tuple from the partition queue for $p$ only when value of \texttt{count[$p$]} (counter for $p$) is zero before the atomic increment. Now, to prove (1), it is enough to assert that the value of \texttt{count[$p$]} is never zero when a tuple belonging to $p$ is being actively processed by a worker.  In the \texttt{do-while} loop (L11-14), the counter is decremented only after the dequeued tuple is processed completely. Note that the control flow in the \texttt{addInput} procedure ensures that \texttt{tryDequeue} at L12 always succeeds. Since the counter is decremented only at L14, it is clear that only the active worker of $p$ can reduce the value to zero, after which any other worker can enter L11-14. The atomic decrement and the condition at L14 ensures that the current active worker does not process any more tuples when \texttt{count[$p$]} becomes zero. Hence, at most one worker operates on tuples belonging to the same key. 
	\par
	Further, the FIFO guarantee of the linearizable concurrent queues in \texttt{partitionQueues} and the constraint that at most one worker can enter L11-14 for a particular $p$ (proved above) ensure that tuples belonging to the same key are processed exactly once and in order of their arrival into \texttt{partitionQueues[$p$]}. 
\end{proof}

\paragraph{Progress} No worker can get blocked in the hybrid-queue approach. Adding inputs  happen only by a single worker (of the operator upstream) due to the execution model employed in the reordering scheme. When consuming inputs, a worker that dequeues a tuple with same key as one being concurrently processed by another worker will simply delegate it to the active worker. So, this worker does not get blocked and moves on to the next key in the master queue. In this approach, outputs are also produced almost in their arrival order, which avoids blocking of outputs (sometimes workers themselves) by the reordering scheme.
\section{Correctness of Implementation}
\label{sec:correctness} 
In this section, we describe how we use the concurrent data structures described in Sec. \ref{sec:reordering-scheme} and \ref{sec:partitioned-parallelism} to implement the operators and prove that our implementation in combination with the runtime always guarantees ordered processing.
\par
We start by defining the notion of correctness on the concurrent execution resulting from any implementation of the streaming computation. 
\begin{definition}[Ordered execution]
	A concurrent execution $E$ of a streaming computation on any input sequence $i_1, i_2, i_3, ... $ is \emph{ordered}, if and only if, the output sequence $o_1, o_2, o_3, ...$ produced  by the execution is the same output sequence produced by a sequential execution of the pipeline on $i_1, i_2, i_3, ...$. 
\end{definition}
We would like to ensure our implementation of the streaming computation is correct with respect to the above definition of ordering. An implementation of a streaming computation is said to be \emph{ordered}, if and only if, any concurrent execution of the implementation is ordered.
\par    
There are three types of operators supported in our system: stateful, stateless and partitioned stateful operators. The implementation of a stateful operator is straight-forward. A worker of the upstream operator adds an input tuple to its worklist(a single-producer single-consumer concurrent queue). Only a single worker is allotted to a stateful operator at any time and this worker consumes these inputs serially and adds the corresponding outputs to the worklist of the downstream operator. A stateless operator is built using a shared-worklist (a multi-producer multi-consumer concurrent queue) and our non-blocking reordering buffer (Sec. \ref{sec:reordering-scheme}). Input tuples are added to the shared-worklist and every tuple is allotted  a unique serial number using an atomic counter. Worker(s) allotted to this stateless operator dequeue an input from this worklist and process it to produce output, which are then sent to the downstream operator by invoking the \texttt{send} method of the reordering buffer(fig. \ref{alg:non-blocking-reordering-scheme}). If it fails, the worker tries again until it successfully adds the output to this buffer.
\par
We implement the partitioned stateful operator by composing the hybrid partitioning scheme we described in Sec. \ref{sec:partitioned-parallelism} with our non-blocking reordering buffer. Inputs are allotted a unique increasing serial number in the order of their arrival and added by invoking the \texttt{addInput} method (fig. \ref{alg:hybrid-queue}). Workers alloted to this operator consumes inputs using the \texttt{consumeInputs} method and invokes the \texttt{send} method of our reordering buffer (fig. \ref{alg:non-blocking-reordering-scheme}) to send outputs downstream. 

\begin{theorem}[Correctness of pipeline implementation]
	Any pipeline built by composing the above operator implementations and executed using our dynamic runtime only allows ordered executions.
\end{theorem}
\begin{proof}
It is easy to see that the above theorem holds for a pipeline composed only of a single stateful operator. For a pipeline composed only of a stateless operator, even though $i_1, i_2, i_3 ...$ may be processed in any order, the corresponding output produced for $i_t$ is $o_t$ since the operator is stateless. The reordering buffer (sec. \ref{sec:reordering-scheme}) guarantees that reordered sequence sent out is $o_1, o_2, o_3 ...$. 
\par
Now, let us consider a pipeline composed only of a partitioned stateful operator. For any two inputs $i_k$ and $i_l$ $(k < l)$, if they belong to different keys, then irrespective of the order in which they are processed the corresponding outputs produced will be $o_k$ and $o_l$. When they have the same keys, the hybrid scheme guarantees that $i_k$ will be processed before $i_l$ and hence the outputs produced will be $o_k$ and $o_l$. So, the output produced for $i_t$ is $o_t$. Similar to a stateless operator pipeline, reordering scheme ensures that the output sequence produced is $o_1, o_2, o_3,...$.
\par
Since each of these single operator pipelines lead to correct executions, it is straightforward to see any linear composition of these operators will always lead to correct executions in our runtime.
\end{proof}

\section{Dynamic Scheduling}
\label{sec:dynamic-scheduling}
Our system consists of many workers that consume inputs from the worklist(s) of an operator to produce outputs using the user-specified operator logic. The number of such workers is the same as number of cores available on the multicore machine. Each worker queries a central scheduler data structure to obtain some work and returns back for more, after it finishes the work allotted previously. The scheduler is responsible for answering two questions: (1) which operator to work on and (2) how many tuples to process from its worklist(s) before returning back. In this section, we propose some scheduling heuristics to perform this dynamic work allotment to worker threads.
\par
We say an operator is \emph{schedulable}, if the currently allotted number of workers is less than its maximum allowed degree of parallelism and its worklist is not empty. The theoretical maximum degree of parallelism of a stateful operator is 1, of a partitioned stateful operator is the number of partitions $p$, and that of a stateless operator is $\infty$ (essentially the number of available cores $n$). In all the heuristics we discuss below, we consider only those operators that are schedulable at the time we make the scheduling decision.
\par 
A worker when allotted to an operator, operates on it for a constant time slice $s$. The maximum number of tuples that must be processed by the worker can be computed using this constant $s$ and $c_i$, the cost of processing a single input tuple by $o_i$. If the worklist of an allotted operator becomes empty before processing the specified number of tuples, the worker does not get blocked; instead returns back to query the scheduler for more work.There are several alternatives for choosing the time for which an operator should be scheduled. However, we focus on constant time slices in order to study characteristics of the heuristics we propose without interference from these changes. Nevertheless, one has to be careful in choosing $s$. Higher the value of $s$, lower the contention for querying the scheduler and better amortization of scheduling overheads. On the other hand, a larger value of $s$ impedes the responsiveness of the system to dynamic changes, as it can get stuck on a previous scheduling decision for a long time.
\par
We first propose some intuitive heuristics based on the idea of orchestrating the flow of tuples through a pipeline. There are two simple ways to enable this flow: one is to provide a thrust from ingress towards egress or use a suction pressure from egress to pull items from ingress. The following two heuristics are based on this key idea.
\subsection{Queue-size-throttling (QST)} 
In this heuristic, we push tuples from the entry point towards the exit point and try to focus on one operator at a time. We schedule an operator until it generates enough inputs for the downstream operators and then go on to schedule the next one in the pipeline. We implement this scheme using queue throttling: each operator has an upper bound on its output queue (worklist of the downstream operator) size and is not scheduled if current size is higher than this threshold. In short, the heuristic always picks the earliest operator in the pipeline that has current output queue size less than its threshold. 
\par
Further, each operator $o_i$ has a selectivity, denoted by $s_i$, which is the average number of outputs produced by $o_i$ on processing a single input tuple. 
For example, selectivity is 1 for a \texttt{map} operator that maps each input tuple to a single output tuple, while it is less than 1 for a \texttt{filter} and more than 1 for \texttt{flat-map}, which maps a single input tuple to more than one output tuples. Due to difference in selectivities, having a uniform threshold for all operators could potentially create a slack in the pipeline. So, we set the output queue size threshold $T_i$ for an operator $o_i$ as follows, where $cs_i$ is the cumulative selectivity of operator $o_i$ since ingress $(cs_i = \prod_{k=1}^{i} s_k)$ and $C$ is a constant that can be imagined as capacity of the system. 
\begin{align}
	T_i  &= \frac{C * cs_i}{\sum_{i = 1}^{n} cs_i}
\end{align}
Note that $T_i$ is proportional to the expected number of tuples produced by $o_{i}$ as input to $o_{i+1}$, when $\frac{C}{\sum_{i = 1}^{n} cs_i}$ tuples are processed in the overall pipeline.
\subsection{Last-in-pipeline (LP)} 
This heuristic is based on the complementary idea of pulling tuples from the exit point. In contrast to QST, this heuristic seeks to schedule operators later in the pipeline. Whenever an operator is not schedulable, this heuristic moves to its upstream operator and schedules that. This scheme depends entirely on the imminent dataflow between the operators and not on any of the operator characteristics. So, LP chooses the latest operator in the pipeline that has a non-empty input queue. An alternative could be to have a minimum worklist size, in which case only operators with worklist at least as big as this threshold would be considered for scheduling. But, in our empirical evaluation we consider only the simpler case where this threshold is 1.
\par
The next set of heuristics take a slightly different approach to scheduling by \emph{prioritizing} operators based on a certain measure of priority. This priority is computed using operator characteristics and current status of the pipeline. Essentially, these heuristics answer the question: which operator in the pipeline currently needs the most worker time to reach our performance goals? We discuss two heuristics designed using this strategy below. 
\subsection{Estimated-time (ET)} 
In this heuristic, we prioritize operators based on the estimated time it would take to process its current worklist, if we allot a new worker to it. We compute priority $p_i$ of an operator $o_i$, as follows, where $I_i$ denotes the current size of its worklist, $c_i$ denotes the cost of processing a single tuple by $o_i$, $w_i$ denotes the number of workers currently assigned to $o_i$, and $M_i$ is its maximum allowed degree of parallelism:
\begin{equation}
 p_i = 
 \begin{cases}
   \frac{I_i * c_i}{w_i + 1} & \text{if } w_i < M_i \\
   0 & \text{otherwise}
 \end{cases}
\end{equation}
This strategy is based on the intuition that an operator that needs more worker time will lag behind and have a worklist that will take longer to complete.
\subsection{Current-throughput (CT)} 
The key idea here is to choose the operator with the lowest throughput, as it is likely to be the \emph{bottleneck} in the pipeline.  We have to normalize the throughput to account for non-unit selectivities. We divide the time dimension into windows of size $w$ and compute the effective number of tuples processed by an operator in that time window as a measure of its throughput. The effective number of tuples $n^w_i$ that would be processed in the current window under current allocation of workers can be computed  approximately as follows:
\begin{equation}
	n^w_i = \frac{T^w_i + w_i * s}{c_i * cs_i}
\end{equation}
where, $T^w_i$ is the total worker time spent on $o_i$ in the current window $w$, $w_i$ is the number of workers alloted to $o_i$ currently and $s$ is the time slice for which each of the $w_i$ workers are allotted to $o_i$. CT chooses the operator with the lowest $n^w_i$ value.  Another critical issue in the above heuristic is deciding on the window size $w$. It is possible for the scheduler to make sub-optimal decisions if the window size $w$ is too low. Ideally, we would like to use a window size that would have same $n^w_i$ for all the operators at the end of the window. This is similar to the \emph{period} of a static schedule.
\par
We evaluate these heuristics on real-world streaming queries from the TPCx-BB benchmark and discuss the pros and cons of choosing one over another in the next section.

\section{Evaluation}
\label{sec:evaluation}
In this section, we present results of evaluation of the different scheduling heuristics and highlight benefits of our design of the parallelization framework for ordered stream processing empirically. 

\paragraph{Experimental Setup} We perform all our experiments on Intel Xeon E5 family 2698B v3 series which runs the Windows Server 2012 R2 Datacenter operating system. 
It has 16 physical cores, with L1, L2 and L3 cache of size 32 KB, 256 KB and 40 MB. We implemented our research prototype in C++ on Windows using standard library implementations of concurrent queues and other atomic primitives. We measure throughput and latency by sending marker wrappers over tuples at equal tuple intervals, which carry information about entry and exit times. We ran all experiments for 2-10 mins and report the mean over 3 runs. For measurements, we consider only markers in the 20th to 80th percentile range, to eliminate starting up and shutting down interferences. Average throughput is computed by obtaining the ratio of number of tuples to the total time taken to process them and latency by averaging the processing latency of each marker in the range.

\paragraph{Benchmark} We use queries from the TPCx-BB benchmark\cite{TPCxBB}, which is a modern Big Data benchmark that covers various categories of data analytics. We use all queries (Q1, Q2, Q3, Q4 and Q15) that correspond to stream processing workloads from TCxBB to compare our heuristics and evaluate various aspects of the runtime design. These queries and their implementation details are summarized in table \ref{query-summary}. Web clickstreams are generated by every click made by a user on the online shopping portal and every item purchase in a retail store generates a store sales tuple. 

\begin{table*}
	\centering
	\begin{tabular}{ | l | l | p{11.5cm}|}
		\hline
		 &\textbf{Pipeline} &  \textbf{Brief Description} \\\hline
		1 & SS  $\rightarrow$  SL $\rightarrow$ PS $\rightarrow$ PS $\rightarrow$ SF& Find top 100 pairs of items that are sold together frequently in the retail stores every hour\\\hline
		2 & WC  $\rightarrow$  SL $\rightarrow$ PS $\rightarrow$ SL $\rightarrow$ PS $\rightarrow$ SF & Find top 30 products that are viewed together online. Viewed together relates to a click-session of a user with session time-out of 60 mins\\\hline
		3 & WC  $\rightarrow$   SL $\rightarrow$ PS $\rightarrow$ PS& Find top 30 list of items (sorted by number of views) which are the last 5 products (in the past 10 days) that are mostly viewed before an item was purchased online\\\hline
		4 & WC  $\rightarrow$   SL $\rightarrow$ PS $\rightarrow$ SL $\rightarrow$ SF & Shopping cart abandonment analysis: For users who added products in their shopping cart but did not check out, find average number of pages they visited during their session\\\hline
		15 & SS  $\rightarrow$   SL $\rightarrow$ SL $\rightarrow$ PS & Find item categories with flat or declining sales for in-store purchases\\\hline
	\end{tabular}
	\caption{Summary of streaming queries in TPCx-BB. In the above table, WC = web clickstreams, SS = store sales, SL = stateless, PS = partitioned stateful and SF = stateful operator}
	\label{query-summary}
\end{table*}

\subsection{Comparison of scheduling heuristics}
\begin{figure*}
\begin{subfigure}[b]{0.49\textwidth}
	\centering
	\includegraphics[scale=0.5]{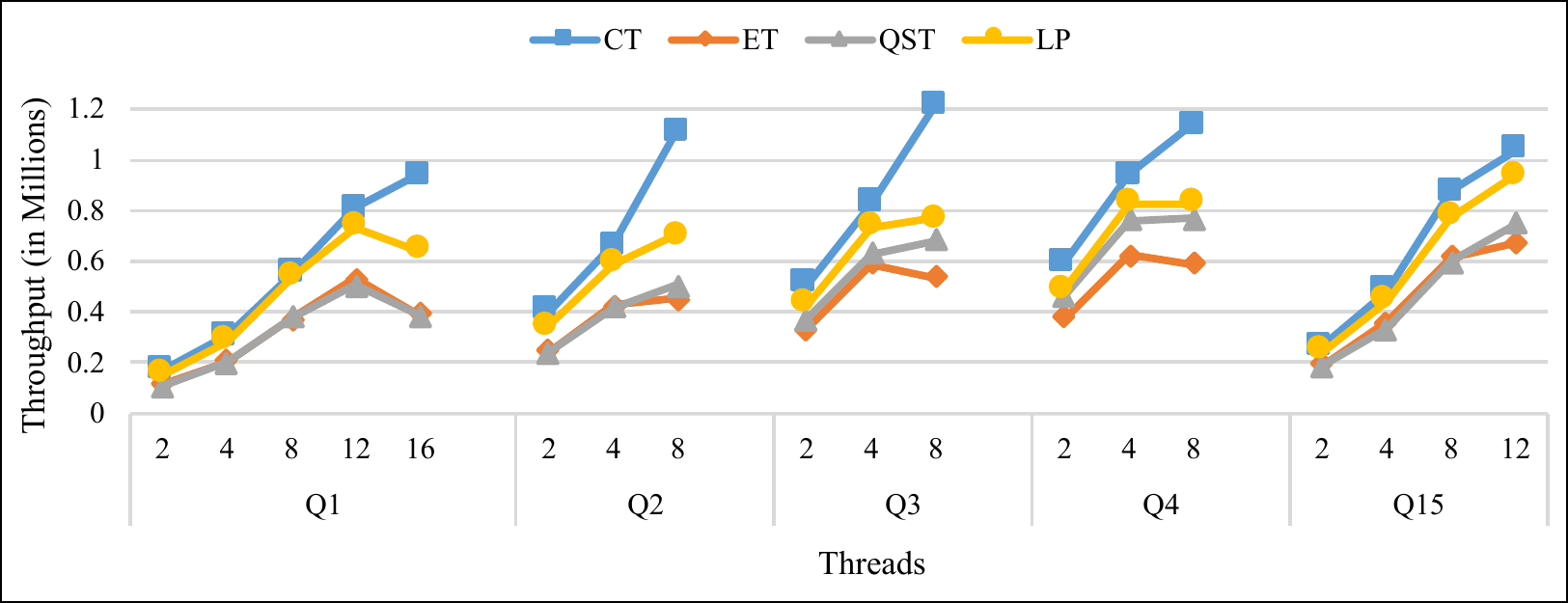}
    \caption{Average Throughput}
	\label{fig:bigbench-schedulers-throughput}
\end{subfigure}
\begin{subfigure}[b]{0.49\textwidth}
	\centering
	\includegraphics[scale=0.5]{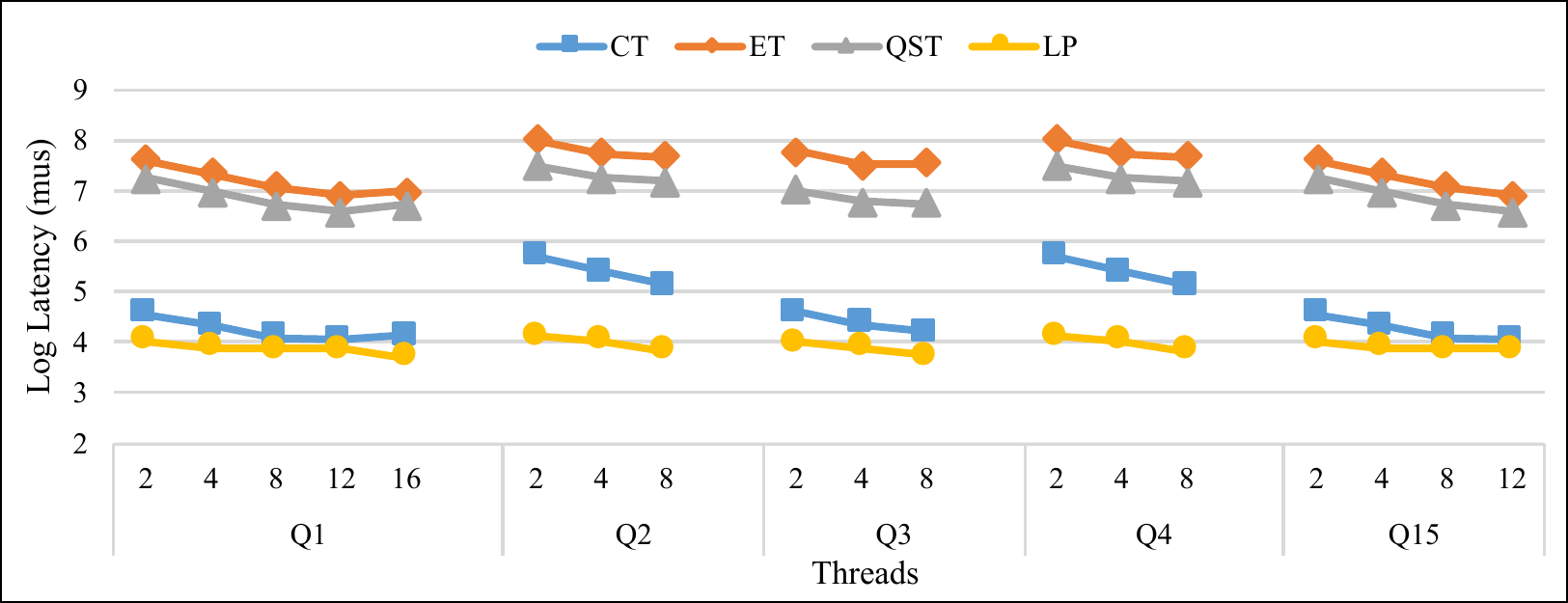}
    \caption{Average latency}
	\label{fig:bigbench-schedulers-latency}
\end{subfigure}
\caption{ Performance of the runtime when using different scheduling heuristics over TPCx-BB queries (We increase number of cores until peak throughput of the best heuristic - performance drops after that due to overheads of parallelization.)}
\end{figure*}

We discussed several heuristics for dynamically scheduling operators in a stream processing pipeline in Sec. \ref{sec:dynamic-scheduling}, namely normalized-current-throughput (CT), estimate-completion-time (ET), last-in-pipeline (LP) and queue-size-throttling (QST).  
We present results of experiments comparing their performance in terms of throughput and latency for the above queries in figures \ref{fig:bigbench-schedulers-throughput} and \ref{fig:bigbench-schedulers-latency} respectively. 
We increase the number of cores until peak throughput of the best heuristic, beyond which the performance drops when the overheads of parallelization outweigh its benefits. 
We can use existing techniques in the literature~\cite{IBMStreams} to identify this break-even point automatically. So, we do not focus on that aspect here.
\par
\textbf{Throughput.} We observe that heuristic CT scales almost linearly up to 16 cores for Q1, 12 cores for Q15 and 8 cores for Q2, Q3 and Q4. 
It achieves a peak throughput of approximately millions of web clickstreams and store sales tuples per second. 
We observe that this is the best possible throughput based on the per-tuple processing costs and selectivities of operators in the pipeline for corresponding degrees of parallelism. 
Among other heuristics, LP performs as well as CT for queries Q1 and Q15, but achieves sub-optimal performance for the others.
Both ET and QST are observed to follow a similar trend in speedup achieved, however, they do not perform as well as CT or LP in terms of absolute throughput.
\par 
\textbf{Latency.} LP is the best heuristic for low-latency processing, followed closely by CT. 
It achieves latencies as low as a few milliseconds, which is the best known for stream processing systems. 
CT, which yields the best throughput, also processes tuples with such low latencies in many cases while it shoots up to 100s of milliseconds in some cases. 
Note that this is still quite low compared to other stream processing systems, which are based on batched stream processing \cite{SparkStreaming, Trill}. 
On the other hand, ET and QST have quite high latencies. 
This increase in latency for QST maybe due to a higher value of $M$ (refer Sec. \ref{sec:dynamic-scheduling}), while ET is heavily influenced by the throughput of input stream to the overall pipeline.
\par
\textbf{Analysis.} From our analysis of the experimental results, we observe that there is a difference in performance among the heuristics even when their worker time distribution (ratio of total worker time spent on each operator in the pipeline) is almost similar. 
Heuristics that distribute workers across operators in the pipeline simultaneously tend to establish a continuous pipelined flow and are seen to yield much better throughput and latency. 
Those that focus on a single operator by exploiting maximum data parallelism at a time lead to increased per-tuple processing cost due to overheads at higher degrees of parallelism. 
\par
CT and LP seem to be exploiting this dichotomy quite efficiently. 
Choosing the operator with lowest estimated normalized throughput in the current window easily establishes this pipelined flow and hence uses an ideal combination of data, partitioned and pipeline parallelism.
LP, that aims to always schedule operators later in the pipeline also establishes this continuous flow as follows: Initially, it is forced to schedule earlier operators in the pipeline as later ones are not schedulable; as they are scheduled it generates inputs for later operators and any worker that exits this operator is scheduled immediately on the next while some others are still processing the earlier operator. However, LP over-allots workers to operators later in the pipeline when they are schedulable which leads to sub-optimal performance in some queries above. 
The QST heuristic focuses on one operator at a time by design, similar to batched stream processing, thereby scheduling operators one-by-one along the pipeline. 
ET seems to be highly influenced by the input stream throughput as priority of the first operator depends on this. 
Hence, for a value of throughput higher than current system throughput, ET focuses mainly on the earliest operator and leads to sub-optimal performance as is evident from the results.
\par
In the next two sub-sections, we discuss certain aspects of our parallelization framework that handles concurrent workers allotted to the same data or partitioned parallel operator. 
We designed \emph{parametric operators} that can be used to create stateless and partitioned stateful operators with different computation profiles to help analyze their scalability in our framework. 
These operators are based on matrix computations on the input tuple. 
The per-tuple processing cost, input tuple size, state size (for partitioned stateful) and selectivity can be varied by initializing these parametric operators with appropriate parameters.  

\subsection{Comparison of Partitioning Schemes}
Now, we compare the two partitioning schemes we described in Sec. \ref{sec:partitioned-parallelism}, PARTITIONED-QUEUE and HYBRID-QUEUE, that help achieve partitioned parallelism. Specifically, we compare their performance during load imbalance and in terms of latency with the constraint of ordered processing. Both schemes behave similarly in terms of per-operator throughput under uniform distribution, but hybrid scheme performs better in longer pipeline queries as it is more amenable to pipeline parallelism.

\subsubsection{Load Balancing}
\begin{figure}
	\centering
	\includegraphics[scale=0.55]{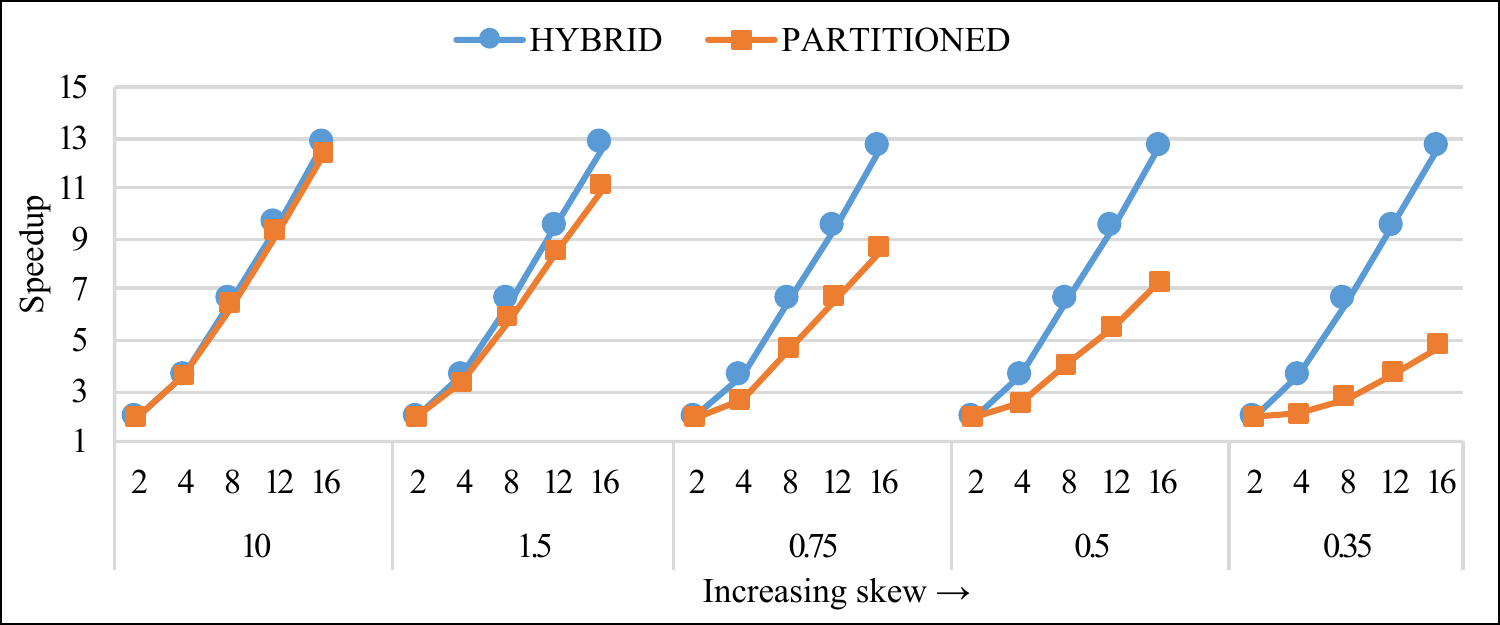}
	\caption{Handling load imbalance across partitions}
	\label{fig:partitioning-schemes-load-imbalance}
\end{figure}

\begin{figure}
	\centering
	\includegraphics[scale=0.55]{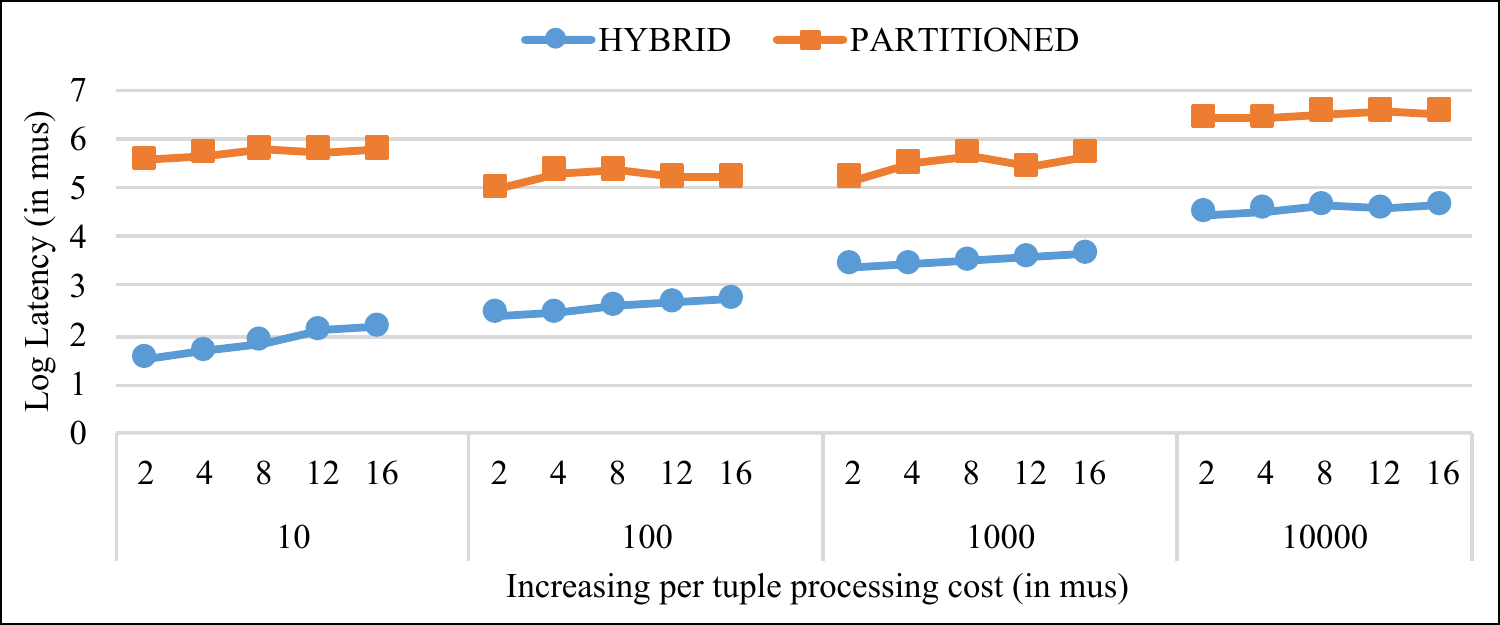}
	\caption{Latency vs. Per-Tuple Processing Cost}
	\label{fig:partitioning-schemes-latency}
\end{figure}
Skewed distribution is known to highly limit partitioned parallelism. It is especially important to be able to balance load across workers in the stream processing setting as they are expected to be long running continuous queries. In this experiment, we provide empirical evidence that the HYBRID-QUEUE approach can handle load imbalance much better than the PARTITIONED-QUEUE approach. In order to systematically induce skew in the distribution, we do range partitioning on  keys sampled from a Gaussian distribution. We scale values in $[-1, 1]$ generated by $\mathcal{N}(0, \sigma)$ appropriately to fit the key space. We vary the value of $\sigma$ to vary the skew across partitions - higher the value of $\sigma$, closer the distribution is to a uniform distribution. The maximum number of partitions for PARTITIONED-QUEUE is limited to the number of workers, while the number of partitions in the HYBRID-QUEUE allows finer partitions and so is set to 100. The results of this experiment are presented in figure \ref{fig:partitioning-schemes-load-imbalance}. We observe that both schemes perform similarly when the distribution is almost uniform. However, as we increase skew in the distribution, HYBRID-QUEUE performs consistently while scalability of the PARTITIONED-QUEUE approach drops heavily. This is because HYBRID-QUEUE admits finer partitions and hence leads to better load balancing.

\subsubsection{Latency} 
We now compare the average processing latency in either schemes - the time between start of processing an input to the time at which its outputs exit the operator through the reordering scheme. We observe this for operators with various per-tuple processing costs (10, 100, 1000 and 10000 micro seconds) and a uniform distribution of tuples across partitions - the results are presented in figure \ref{fig:partitioning-schemes-latency}. We can see that the average processing latency is much higher for PARTITIONED-QUEUE, while for HYBRID-QUEUE it is close to the corresponding operator's per-tuple processing cost. This is because the outputs produced through the PARTITIONED-QUEUE approach has to wait longer in the reordering buffer for outputs with earlier serial numbers. We do not report throughput comparisons between the two schemes here as both yield similar throughputs due to a uniformly random distribution of keys. However, this difference in individual operator processing latency leads to throughput differences in larger pipeline queries as we will see in the next experiment.

\subsubsection{Pipeline queries}
\begin{figure}
	\centering
	\includegraphics[scale=0.55]{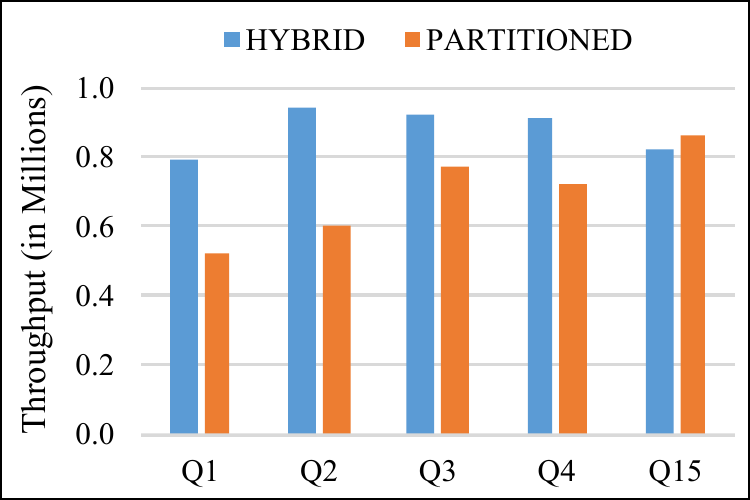}
	\includegraphics[scale=0.55]{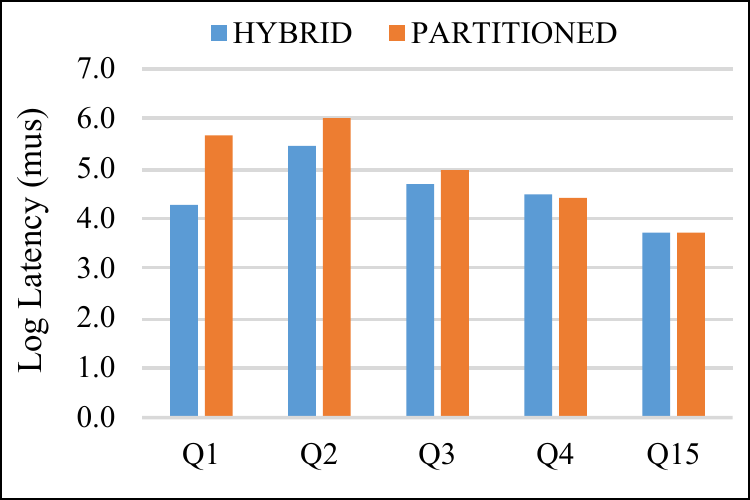}
	\caption{Peak throughput (left) and latency (right) of the two partitioning schemes on TPCx-BB queries}
	\label{fig:partitioned-scheme-bigbench}
\end{figure}
We compare performance of the two approaches on the TPCx-BB queries we described above. We use the CT scheduling heuristic, which yields the best performance among all the heuristics, and change only the partitioning scheme keeping the rest of the framework same. Peak throughput and latencies are reported in figure \ref{fig:partitioned-scheme-bigbench} as we vary the number of workers from 2 to 16. HYBRID-QUEUE is able to achieve much higher throughput than PARTITIONED-QUEUE in all queries. As expected, the difference is higher for queries that have more partitioned stateful operators. Query 15 contains only one such operator and is partitioned on item category id, where the total number of item categories in TPCx-BB is 10. It does not support higher degrees of parallelism and so the difference is unclear. HYBRID-QUEUE performs better than PARTITIONED-QUEUE also in terms of latency in 3 out of 5 queries, which have more partitioned stateful operators and almost similar for the rest. 

\subsection{Comparison of Reordering Schemes}
We report the results of our empirical evaluation comparing the NON-BLOCKING scheme (fig. \ref{alg:non-blocking-reordering-scheme})  and the LOCK-BASED scheme (fig. \ref{alg:naive-reordering-scheme}) in this subsection. We specifically highlight scenarios which are seen to be important in real-world queries from TPCx-BB using micro-benchmark experiments and also support it by evaluating them on the pipeline queries themselves.

\subsubsection{Light-weight Operators} \label{sec:reordering-small-op}  When the per-tuple processing cost of a stateless or partitioned-stateful operator is large and its computation profile is amenable to parallelization, the overhead of reordering outputs is relatively smaller and hence does not impede scalability of the operator. However, when this quantity is small, reordering could potentially become a huge bottleneck. We demonstrate that our NON-BLOCKING strategy minimizes this overhead leading to better scalability of such operators.  We designed a stateless parametric operator with a per-tuple processing cost in the order of 10s of microseconds on a single core serial execution. Now, we varied the degree of parallelism of this operator and observed the increase in average per-tuple processing cost and the corresponding speedup achieved (fig. \ref{fig:reordering-small-op}). Higher the reordering overhead, higher the average per-tuple processing cost and lower the speedup achieved. The results show that NON-BLOCKING reordering scheme scales better than  the LOCK-BASED scheme. As expected, the average per-tuple processing cost of the operator, which includes the time for which a worker is blocked, increases more steeply for the LOCK-BASED strategy due to unnecessary blocking of workers when another worker is sending outputs downstream. This is avoided in our improved non-blocking design.

\begin{figure}
	\centering
	\includegraphics[scale=0.55]{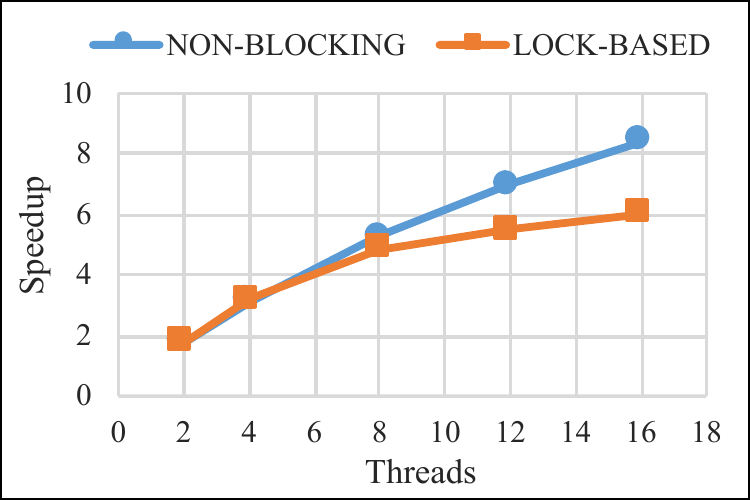}
	\includegraphics[scale=0.55]{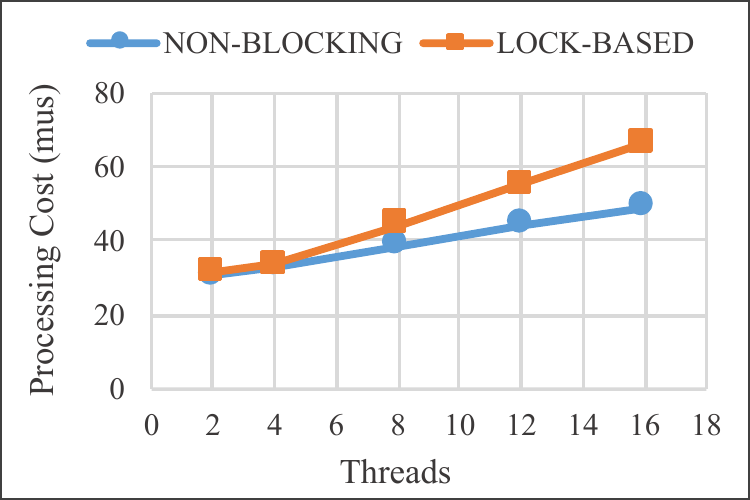}
	\caption{(a) Speedup  and (b) average processing cost for a light-weight operator}
	\label{fig:reordering-small-op}
\end{figure}

\begin{figure}
	\centering
	\includegraphics[scale=0.55]{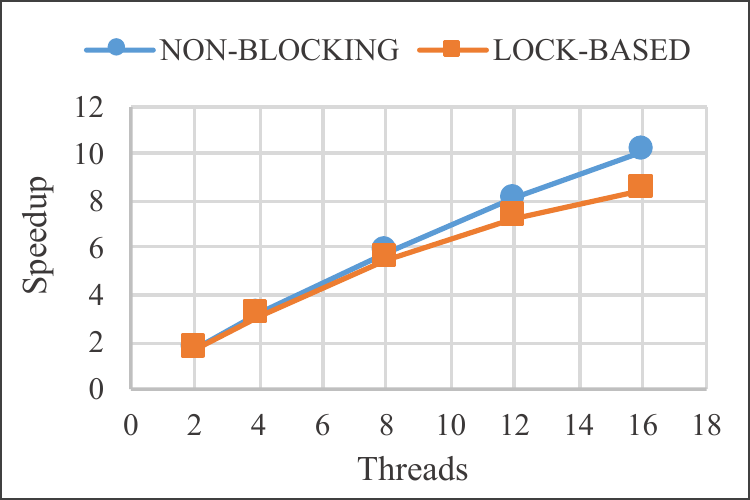}
	\includegraphics[scale=0.55]{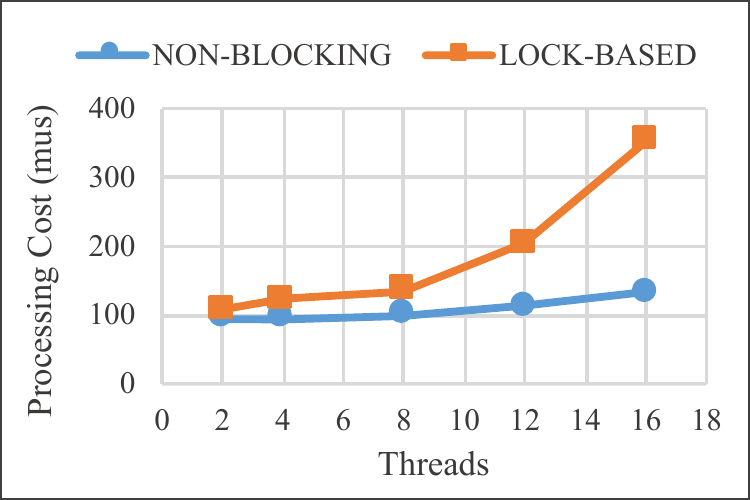}
	\caption{(a) Speedup  of the pipeline and (b) average processing cost of the first operator with LOCK-BASED and NON-BLOCKING reordering schemes}
	\label{fig:selectivity-pipeline}
\end{figure}
\subsubsection{High Selectivity Operators}
Similarly, when these operators have a huge selectivity (number of outputs per input),  the amount of serial overhead involved in reordering is higher. In such cases, NON-BLOCKING strategy performs better in comparison to LOCK-BASED, even for operators with larger computation sizes . To illustrate this, we construct a pipeline that consists of two operators, a parametric stateless operator that is followed by a partitioned stateful operator. We use operators with a processing cost of approximately 100$\mu$s and the stateless operator has a selectivity of 50. Such high selectivity is not uncommon in real workloads. For example in query Q2, all clickstreams in a session are analyzed to produce a large number of item pairs viewed together. We report the average per-tuple processing cost and speedup achieved for this pipeline query in figure \ref{fig:selectivity-pipeline}. Every tuple in the batch of outputs generated by the stateless operator has to be added into the appropriate queue of the partitioned stateful operator. To ensure ordering constraints, this operation is performed serially, which leads to blocking of workers in LOCK-BASED strategy, while in our scheme this is avoided. 
\begin{figure}
	\centering
	\includegraphics[scale=0.7]{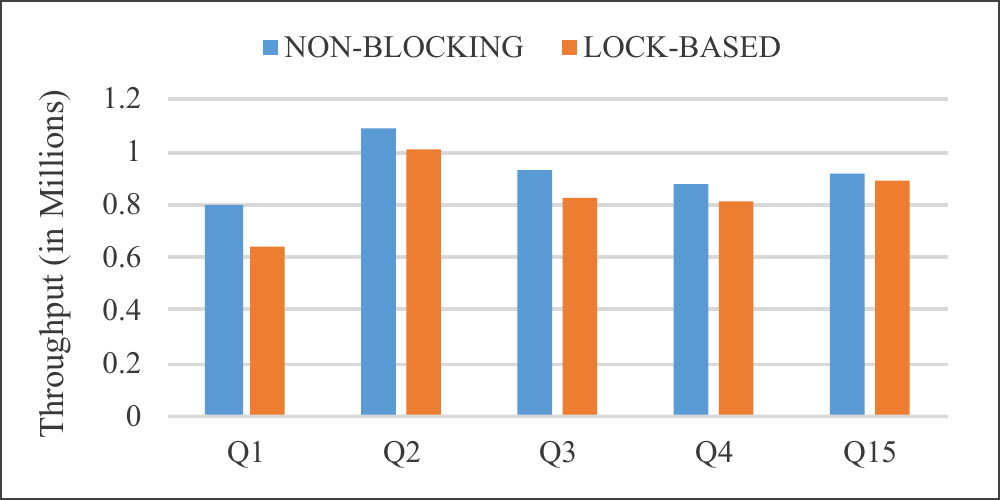}
	\caption{Peak throughput comparison of the reordering schemes on TPCx-BB queries}
	\label{fig:reordering-scheme-bigbench}
\end{figure}

\subsubsection{Pipeline Queries}
To further validate the benefits of our non-blocking reordering scheme, we compare it against the LOCK-BASED scheme on TPCx-BB queries. We report the peak throughput of the runtime for each of the queries using the best heuristic (CT) and by varying just the reordering scheme in fig. \ref{fig:reordering-scheme-bigbench}. We can clearly see that NON-BLOCKING scheme consistently yields a better throughput than the LOCK-BASED strategy. They do not differ much in processing latency  and hence we do not present them here. 

\section{Related Work}
\label{sec:related-work}
In this section, we review prior work related to concurrent data structures that we designed for ordered processing and scheduling of streaming computations. 

\subsection{Concurrent Data Structures}
Our system uses several non-blocking concurrent data structures that have been proposed in the literature~\cite{MauriceHerlihyBook} such as single-producer single-consumer FIFO queues and multi-producer, multi-consumer queues. The reordering scheme we presented has very specific requirements (non-blocking, low-latency buffering), which are not directly met by any other data structure. The pre-allotted monotonically increasing serial numbers enabled further optimizations that would be inaccessible to a generic data structure such as concurrent priority queues. 
\par
Most partitioned parallelism implementations are based on the partitioned-queue approach we presented in Sec. \ref{sec:partitioned-parallelism}, initially proposed in the Volcano~\cite{Volcano} model of query evaluation for databases. In such a design, the degree of parallelism associated with the operator is determined statically and cannot be controlled by a dynamic scheduler. In case of shared-nothing architectures, some techniques~\cite{Flux2003:ICDE} exist that adaptively repartitions the query during runtime.  The trade-offs with respect to communication and repartitioning overheads are very different in a shared-memory architecture, so those techniques do not apply here directly. In addition, we address partition parallelism in the presence of ordering constraints, which to the best of our knowledge, none of the existing concurrent data structures address. 

\subsection{Static Scheduling} Static schedulers assume that the per-tuple processing cost and selectivity of the operators are known at compile time. Early streaming systems designed for applications from the digital signal processing domain focused on compiling down synchronous dataflow graphs (SDF), to single and multicores \cite{Battacharyya:1996, Battacharyya:1999}. For their application domain a purely static solution is not unreasonable as operator characteristics are largely fixed. There is a huge body of literature on scheduling SDF graphs to optimize various metrics such as throughput, memory and cache locality \cite{SDF1, SDF2, SDF3}. StreamIt \cite{StreamIt}, Brook \cite{Brook} and Imagine \cite{Imagine} are some of the early systems designed based on this model of execution. However, none of these works address the case when operator characteristics change during runtime. 

\subsection{Dynamic Solutions}
Aurora \cite{Aurora}, its distributed counterpart, Borealis \cite{Borealis} and STREAM \cite{STREAM} are some of the early prototypes of stream processing engines that make dynamic scheduling decisions. Many recent stream processing engines (NaiagraST \cite{NiagraST}, Nile \cite{Nile}, Naiad \cite{Naiad},Spark Streaming\cite{SparkStreaming}, Storm \cite{Storm}, S4 \cite{S4}) also scheduling decisions during runtime. All these systems either focus on single core or shared-nothing architectures. Even distributed solutions composed of individual shared-memory multicores consider each core as a separate executor and hence fail to exploit the advantages of a fast shared-memory. 
\par
IBM Streams~\cite{IBMStreams} is one of the systems that target shared-memory architecture. Their runtime focuses on two issues: First, they design a mechanism to dynamically determine the maximum number of cores needed by a pipeline. This work is orthogonal to our work and can be easily adapted to our system. Second, they design a scalable scheduler that can schedule a large pipeline on a multicore; the focus is on scalability of the scheduler (number of scheduling decisions made) and not necessarily overall performance of the pipeline as they assume manual fine-tuning. Another key difference is that their scheduler works on an expanded pipeline, where each logical operator is duplicated a number of times specified through user annotations. This limits the flexibility of scheduler while also increasing the scheduling overhead. Their system also does not natively support totally ordered processing making a direct comparison infeasible.   
\par
Other systems such as TRILL~\cite{Trill} and Spark Streaming~\cite{SparkStreaming} are based on the micro-batch architecture.  The idea is to execute a batch of inputs on an operator to completion before starting the next operator, thus relying primarily on the (data) parallelism within an operator. At any given time, a bulk of the workers are involved in executing instances of a single operator. Batching of streams is known to increase latency. We believe that such systems can be built on top of our parallelization and scheduling framework without much effort. We also note that several architectural proposals~\cite{CacheConscious,  RaftLib, AutoPipe, OnTheFly} exist in the literature for a shared-memory streaming parallelization framework, but none of them address dynamic scheduling in the ordered setting or compare different scheduling heuristics empirically, which is a key contribution of this paper. 
\par
The approach we present in this paper is based on dynamic scheduling. Process/thread scheduling in operating systems is an example of this type of scheduling. We seek to develop a customized solution for the streaming setting taking advantage of the extra information available in the form of a dataflow graph. Further in a typical task graph, total amount of work to be done is fixed and the scheduler just needs to pick the right order once whereas in a streaming setting the scheduler has to continuously choose based on the status of pipeline. So, classical notions like work stealing do not apply to our setting \cite{IBMStreams}.
\section{Conclusions and Future Work}
\label{sec:conclusion}
We presented a new design for a dynamic runtime that executes streaming computations on a shared-memory multicore machine, along with the guarantee of ordered processing of tuples. We empirically demonstrated that our runtime is able to achieve good throughput (in millions of tuples per second) without compromising on the latency (a few milliseconds) on some TPCx-BB queries. We presented a couple of concurrent data structures that help achieve data and partitioned parallelism in the ordered setting, proved their correctness and showed their usefulness empirically using micro-benchmarks and on TPCx-BB queries. 
\par
In our current scheme, we assume all worker threads are uniform. However, in reality a worker is closer to some workers than others due to the hierarchical cache architecture and more so in modern non-uniform memory access (NUMA) architectures. An important extension to our work is to design scheduling heuristics that discriminate workers based on their spatial distribution.
\bibliographystyle{ACM-Reference-Format}
\bibliography{references}


\begin{thebibliography}{41}


\ifx \showCODEN    \undefined \def \showCODEN     #1{\unskip}     \fi
\ifx \showDOI      \undefined \def \showDOI       #1{#1}\fi
\ifx \showISBNx    \undefined \def \showISBNx     #1{\unskip}     \fi
\ifx \showISBNxiii \undefined \def \showISBNxiii  #1{\unskip}     \fi
\ifx \showISSN     \undefined \def \showISSN      #1{\unskip}     \fi
\ifx \showLCCN     \undefined \def \showLCCN      #1{\unskip}     \fi
\ifx \shownote     \undefined \def \shownote      #1{#1}          \fi
\ifx \showarticletitle \undefined \def \showarticletitle #1{#1}   \fi
\ifx \showURL      \undefined \def \showURL       {\relax}        \fi
\providecommand\bibfield[2]{#2}
\providecommand\bibinfo[2]{#2}
\providecommand\natexlab[1]{#1}
\providecommand\showeprint[2][]{arXiv:#2}

\bibitem[\protect\citeauthoryear{Abadi, Ahmad, Balazinska, Cetintemel,
  Cherniack, Hwang, Lindner, Maskey, Rasin, Ryvkina, Tatbul, Xing, and
  Zdonik}{Abadi et~al\mbox{.}}{2005}]%
        {Borealis}
\bibfield{author}{\bibinfo{person}{Daniel~J Abadi}, \bibinfo{person}{Yanif
  Ahmad}, \bibinfo{person}{Magdalena Balazinska}, \bibinfo{person}{Ugur
  Cetintemel}, \bibinfo{person}{Mitch Cherniack}, \bibinfo{person}{Jeong-Hyon
  Hwang}, \bibinfo{person}{Wolfgang Lindner}, \bibinfo{person}{Anurag~S
  Maskey}, \bibinfo{person}{Alexander Rasin}, \bibinfo{person}{Esther Ryvkina},
  \bibinfo{person}{Nesime Tatbul}, \bibinfo{person}{Ying Xing}, {and}
  \bibinfo{person}{Stan Zdonik}.} \bibinfo{year}{2005}\natexlab{}.
\newblock \showarticletitle{{The Design of the Borealis Stream Processing
  Engine}}. In \bibinfo{booktitle}{\emph{Second Biennial Conference on
  Innovative Data Systems Research (CIDR 2005)}}. \bibinfo{address}{Asilomar,
  CA}.
\newblock


\bibitem[\protect\citeauthoryear{Abadi, Carney, {\c{C}}etintemel, Cherniack,
  Convey, Lee, Stonebraker, Tatbul, and Zdonik}{Abadi et~al\mbox{.}}{2003}]%
        {Aurora}
\bibfield{author}{\bibinfo{person}{Daniel~J. Abadi}, \bibinfo{person}{Donald
  Carney}, \bibinfo{person}{Ugur {\c{C}}etintemel}, \bibinfo{person}{Mitch
  Cherniack}, \bibinfo{person}{Christian Convey}, \bibinfo{person}{Sangdon
  Lee}, \bibinfo{person}{Michael Stonebraker}, \bibinfo{person}{Nesime Tatbul},
  {and} \bibinfo{person}{Stanley~B. Zdonik}.} \bibinfo{year}{2003}\natexlab{}.
\newblock \showarticletitle{Aurora: a new model and architecture for data
  stream management}.
\newblock \bibinfo{journal}{\emph{{VLDB} J.}} \bibinfo{volume}{12},
  \bibinfo{number}{2} (\bibinfo{year}{2003}), \bibinfo{pages}{120--139}.
\newblock
\urldef\tempurl%
\url{https://doi.org/10.1007/s00778-003-0095-z}
\showDOI{\tempurl}


\bibitem[\protect\citeauthoryear{Agrawal, Fineman, Krage, Leiserson, and
  Toledo}{Agrawal et~al\mbox{.}}{2012}]%
        {CacheConscious}
\bibfield{author}{\bibinfo{person}{Kunal Agrawal}, \bibinfo{person}{Jeremy~T.
  Fineman}, \bibinfo{person}{Jordan Krage}, \bibinfo{person}{Charles~E.
  Leiserson}, {and} \bibinfo{person}{Sivan Toledo}.}
  \bibinfo{year}{2012}\natexlab{}.
\newblock \showarticletitle{Cache-conscious Scheduling of Streaming
  Applications}. In \bibinfo{booktitle}{\emph{Proceedings of the Twenty-fourth
  Annual ACM Symposium on Parallelism in Algorithms and Architectures}}
  \emph{(\bibinfo{series}{SPAA '12})}. \bibinfo{publisher}{ACM},
  \bibinfo{address}{New York, NY, USA}, \bibinfo{pages}{236--245}.
\newblock
\showISBNx{978-1-4503-1213-4}
\urldef\tempurl%
\url{https://doi.org/10.1145/2312005.2312049}
\showDOI{\tempurl}


\bibitem[\protect\citeauthoryear{Ahmad, Groote, H\"{o}lzenspies, Stoelinga, and
  Pol}{Ahmad et~al\mbox{.}}{2014}]%
        {SDF1}
\bibfield{author}{\bibinfo{person}{Waheed Ahmad}, \bibinfo{person}{Robert~de
  Groote}, \bibinfo{person}{Philip K.~F. H\"{o}lzenspies},
  \bibinfo{person}{Mari\"{e}lle Stoelinga}, {and} \bibinfo{person}{Jaco van~de
  Pol}.} \bibinfo{year}{2014}\natexlab{}.
\newblock \showarticletitle{Resource-Constrained Optimal Scheduling of
  Synchronous Dataflow Graphs via Timed Automata}. In
  \bibinfo{booktitle}{\emph{Proceedings of the 2014 14th International
  Conference on Application of Concurrency to System Design}}
  \emph{(\bibinfo{series}{ACSD '14})}. \bibinfo{publisher}{IEEE Computer
  Society}, \bibinfo{address}{Washington, DC, USA}, \bibinfo{pages}{72--81}.
\newblock
\showISBNx{978-1-4799-4281-7}
\urldef\tempurl%
\url{https://doi.org/10.1109/ACSD.2014.13}
\showDOI{\tempurl}


\bibitem[\protect\citeauthoryear{Akidau, Balikov, Bekiroglu, Chernyak,
  Haberman, Lax, McVeety, Mills, Nordstrom, and Whittle}{Akidau
  et~al\mbox{.}}{2013}]%
        {MillWheel}
\bibfield{author}{\bibinfo{person}{Tyler Akidau}, \bibinfo{person}{Alex
  Balikov}, \bibinfo{person}{Kaya Bekiroglu}, \bibinfo{person}{Slava Chernyak},
  \bibinfo{person}{Josh Haberman}, \bibinfo{person}{Reuven Lax},
  \bibinfo{person}{Sam McVeety}, \bibinfo{person}{Daniel Mills},
  \bibinfo{person}{Paul Nordstrom}, {and} \bibinfo{person}{Sam Whittle}.}
  \bibinfo{year}{2013}\natexlab{}.
\newblock \showarticletitle{MillWheel: Fault-Tolerant Stream Processing at
  Internet Scale}. In \bibinfo{booktitle}{\emph{Very Large Data Bases}}.
  \bibinfo{pages}{734--746}.
\newblock


\bibitem[\protect\citeauthoryear{Akidau, Bradshaw, Chambers, Chernyak,
  Fernández-Moctezuma, Lax, McVeety, Mills, Perry, Schmidt, and
  Whittle}{Akidau et~al\mbox{.}}{2015}]%
        {Dataflow}
\bibfield{author}{\bibinfo{person}{Tyler Akidau}, \bibinfo{person}{Robert
  Bradshaw}, \bibinfo{person}{Craig Chambers}, \bibinfo{person}{Slava
  Chernyak}, \bibinfo{person}{Rafael~J. Fernández-Moctezuma},
  \bibinfo{person}{Reuven Lax}, \bibinfo{person}{Sam McVeety},
  \bibinfo{person}{Daniel Mills}, \bibinfo{person}{Frances Perry},
  \bibinfo{person}{Eric Schmidt}, {and} \bibinfo{person}{Sam Whittle}.}
  \bibinfo{year}{2015}\natexlab{}.
\newblock \showarticletitle{The Dataflow Model: A Practical Approach to
  Balancing Correctness, Latency, and Cost in Massive-Scale, Unbounded,
  Out-of-Order Data Processing}.
\newblock \bibinfo{journal}{\emph{Proceedings of the VLDB Endowment}}
  \bibinfo{volume}{8} (\bibinfo{year}{2015}), \bibinfo{pages}{1792--1803}.
\newblock


\bibitem[\protect\citeauthoryear{Apache}{Apache}{2016a}]%
        {Samza}
\bibfield{author}{\bibinfo{person}{Apache}.} \bibinfo{year}{2016}\natexlab{a}.
\newblock \bibinfo{title}{Apache Samza}.
\newblock \bibinfo{howpublished}{{\url{http://samza.apache.org/}}}.
  (\bibinfo{year}{2016}).
\newblock
\newblock
\shownote{Accessed: 2017-01-16.}


\bibitem[\protect\citeauthoryear{Apache}{Apache}{2016b}]%
        {Storm}
\bibfield{author}{\bibinfo{person}{Apache}.} \bibinfo{year}{2016}\natexlab{b}.
\newblock \bibinfo{title}{Apache Storm}.
\newblock \bibinfo{howpublished}{{\url{http://storm.apache.org/}}}.
  (\bibinfo{year}{2016}).
\newblock
\newblock
\shownote{Accessed: 2017-01-16.}


\bibitem[\protect\citeauthoryear{Arasu, Babu, and Widom}{Arasu
  et~al\mbox{.}}{2006}]%
        {CQL}
\bibfield{author}{\bibinfo{person}{Arvind Arasu}, \bibinfo{person}{Shivnath
  Babu}, {and} \bibinfo{person}{Jennifer Widom}.}
  \bibinfo{year}{2006}\natexlab{}.
\newblock \showarticletitle{The CQL Continuous Query Language: Semantic
  Foundations and Query Execution}.
\newblock \bibinfo{journal}{\emph{The VLDB Journal}} \bibinfo{volume}{15},
  \bibinfo{number}{2} (\bibinfo{date}{June} \bibinfo{year}{2006}),
  \bibinfo{pages}{121--142}.
\newblock
\showISSN{1066-8888}
\urldef\tempurl%
\url{https://doi.org/10.1007/s00778-004-0147-z}
\showDOI{\tempurl}


\bibitem[\protect\citeauthoryear{Babcock, Babu, Motwani, and Datar}{Babcock
  et~al\mbox{.}}{2003}]%
        {Chain2003}
\bibfield{author}{\bibinfo{person}{Brian Babcock}, \bibinfo{person}{Shivnath
  Babu}, \bibinfo{person}{Rajeev Motwani}, {and} \bibinfo{person}{Mayur
  Datar}.} \bibinfo{year}{2003}\natexlab{}.
\newblock \showarticletitle{Chain: Operator Scheduling for Memory Minimization
  in Data Stream Systems}. In \bibinfo{booktitle}{\emph{Proceedings of the 2003
  ACM SIGMOD International Conference on Management of Data}}
  \emph{(\bibinfo{series}{SIGMOD '03})}. \bibinfo{publisher}{ACM},
  \bibinfo{address}{New York, NY, USA}, \bibinfo{pages}{253--264}.
\newblock
\showISBNx{1-58113-634-X}
\urldef\tempurl%
\url{https://doi.org/10.1145/872757.872789}
\showDOI{\tempurl}


\bibitem[\protect\citeauthoryear{Babu and Widom}{Babu and Widom}{2001}]%
        {STREAM}
\bibfield{author}{\bibinfo{person}{Shivnath Babu} {and}
  \bibinfo{person}{Jennifer Widom}.} \bibinfo{year}{2001}\natexlab{}.
\newblock \showarticletitle{Continuous Queries over Data Streams}.
\newblock \bibinfo{journal}{\emph{SIGMOD Rec.}} \bibinfo{volume}{30},
  \bibinfo{number}{3} (\bibinfo{date}{Sept.} \bibinfo{year}{2001}),
  \bibinfo{pages}{109--120}.
\newblock
\showISSN{0163-5808}
\urldef\tempurl%
\url{https://doi.org/10.1145/603867.603884}
\showDOI{\tempurl}


\bibitem[\protect\citeauthoryear{Battacharyya, Lee, and Murthy}{Battacharyya
  et~al\mbox{.}}{1996}]%
        {Battacharyya:1996}
\bibfield{author}{\bibinfo{person}{Shuvra~S. Battacharyya},
  \bibinfo{person}{Edward~A. Lee}, {and} \bibinfo{person}{Praveen~K. Murthy}.}
  \bibinfo{year}{1996}\natexlab{}.
\newblock \bibinfo{booktitle}{\emph{Software Synthesis from Dataflow Graphs}}.
\newblock \bibinfo{publisher}{Kluwer Academic Publishers},
  \bibinfo{address}{Norwell, MA, USA}.
\newblock
\showISBNx{0792397223}


\bibitem[\protect\citeauthoryear{Beard, Li, and Chamberlain}{Beard
  et~al\mbox{.}}{2015}]%
        {RaftLib}
\bibfield{author}{\bibinfo{person}{Jonathan~C. Beard}, \bibinfo{person}{Peng
  Li}, {and} \bibinfo{person}{Roger~D. Chamberlain}.}
  \bibinfo{year}{2015}\natexlab{}.
\newblock \showarticletitle{RaftLib: A C++ Template Library for High
  Performance Stream Parallel Processing}. In
  \bibinfo{booktitle}{\emph{Proceedings of the Sixth International Workshop on
  Programming Models and Applications for Multicores and Manycores}}
  \emph{(\bibinfo{series}{PMAM '15})}. \bibinfo{publisher}{ACM},
  \bibinfo{address}{New York, NY, USA}, \bibinfo{pages}{96--105}.
\newblock
\showISBNx{978-1-4503-3404-4}
\urldef\tempurl%
\url{https://doi.org/10.1145/2712386.2712400}
\showDOI{\tempurl}


\bibitem[\protect\citeauthoryear{Bhattacharyya, Murthy, and Lee}{Bhattacharyya
  et~al\mbox{.}}{1999}]%
        {Battacharyya:1999}
\bibfield{author}{\bibinfo{person}{Shuvra~S. Bhattacharyya},
  \bibinfo{person}{Praveen~K. Murthy}, {and} \bibinfo{person}{Edward~A. Lee}.}
  \bibinfo{year}{1999}\natexlab{}.
\newblock \showarticletitle{Synthesis of Embedded Software from Synchronous
  Dataflow Specifications}.
\newblock \bibinfo{journal}{\emph{J. VLSI Signal Process. Syst.}}
  \bibinfo{volume}{21}, \bibinfo{number}{2} (\bibinfo{date}{June}
  \bibinfo{year}{1999}), \bibinfo{pages}{151--166}.
\newblock
\showISSN{0922-5773}
\urldef\tempurl%
\url{https://doi.org/10.1023/A:1008052406396}
\showDOI{\tempurl}


\bibitem[\protect\citeauthoryear{Buck, Foley, Horn, Sugerman, Fatahalian,
  Houston, and Hanrahan}{Buck et~al\mbox{.}}{2004}]%
        {Brook}
\bibfield{author}{\bibinfo{person}{Ian Buck}, \bibinfo{person}{Tim Foley},
  \bibinfo{person}{Daniel Horn}, \bibinfo{person}{Jeremy Sugerman},
  \bibinfo{person}{Kayvon Fatahalian}, \bibinfo{person}{Mike Houston}, {and}
  \bibinfo{person}{Pat Hanrahan}.} \bibinfo{year}{2004}\natexlab{}.
\newblock \showarticletitle{Brook for GPUs: Stream Computing on Graphics
  Hardware}. In \bibinfo{booktitle}{\emph{ACM SIGGRAPH 2004 Papers}}
  \emph{(\bibinfo{series}{SIGGRAPH '04})}. \bibinfo{publisher}{ACM},
  \bibinfo{address}{New York, NY, USA}, \bibinfo{pages}{777--786}.
\newblock
\urldef\tempurl%
\url{https://doi.org/10.1145/1186562.1015800}
\showDOI{\tempurl}


\bibitem[\protect\citeauthoryear{Cao, Gowda, Lakshmi, Narasimhadevara, Nguyen,
  Poelman, Poess, and Rabl}{Cao et~al\mbox{.}}{2017}]%
        {TPCxBB}
\bibfield{author}{\bibinfo{person}{Paul Cao}, \bibinfo{person}{Bhaskar Gowda},
  \bibinfo{person}{Seetha Lakshmi}, \bibinfo{person}{Chinmayi Narasimhadevara},
  \bibinfo{person}{Patrick Nguyen}, \bibinfo{person}{John Poelman},
  \bibinfo{person}{Meikel Poess}, {and} \bibinfo{person}{Tilmann Rabl}.}
  \bibinfo{year}{2017}\natexlab{}.
\newblock \bibinfo{booktitle}{\emph{From BigBench to TPCx-BB: Standardization
  of a Big Data Benchmark}}.
\newblock \bibinfo{publisher}{Springer International Publishing},
  \bibinfo{address}{Cham}, \bibinfo{pages}{24--44}.
\newblock
\showISBNx{978-3-319-54334-5}
\urldef\tempurl%
\url{https://doi.org/10.1007/978-3-319-54334-5_3}
\showDOI{\tempurl}


\bibitem[\protect\citeauthoryear{Carney, \c{C}etintemel, Rasin, Zdonik,
  Cherniack, and Stonebraker}{Carney et~al\mbox{.}}{2003}]%
        {AuroraScheduling}
\bibfield{author}{\bibinfo{person}{Don Carney}, \bibinfo{person}{U\u{g}ur
  \c{C}etintemel}, \bibinfo{person}{Alex Rasin}, \bibinfo{person}{Stan Zdonik},
  \bibinfo{person}{Mitch Cherniack}, {and} \bibinfo{person}{Mike Stonebraker}.}
  \bibinfo{year}{2003}\natexlab{}.
\newblock \showarticletitle{Operator Scheduling in a Data Stream Manager}. In
  \bibinfo{booktitle}{\emph{Proceedings of the 29th International Conference on
  Very Large Data Bases - Volume 29}} \emph{(\bibinfo{series}{VLDB '03})}.
  \bibinfo{publisher}{VLDB Endowment}, \bibinfo{pages}{838--849}.
\newblock
\showISBNx{0-12-722442-4}
\urldef\tempurl%
\url{http://dl.acm.org/citation.cfm?id=1315451.1315523}
\showURL{%
\tempurl}


\bibitem[\protect\citeauthoryear{Chandramouli, Goldstein, Barnett, DeLine,
  Platt, Terwilliger, and Wernsing}{Chandramouli et~al\mbox{.}}{2014}]%
        {Trill}
\bibfield{author}{\bibinfo{person}{Badrish Chandramouli},
  \bibinfo{person}{Jonathan Goldstein}, \bibinfo{person}{Mike Barnett},
  \bibinfo{person}{Robert DeLine}, \bibinfo{person}{John~C. Platt},
  \bibinfo{person}{James~F. Terwilliger}, {and} \bibinfo{person}{John
  Wernsing}.} \bibinfo{year}{2014}\natexlab{}.
\newblock \showarticletitle{Trill: {A} High-Performance Incremental Query
  Processor for Diverse Analytics}.
\newblock \bibinfo{journal}{\emph{{PVLDB}}} \bibinfo{volume}{8},
  \bibinfo{number}{4} (\bibinfo{year}{2014}), \bibinfo{pages}{401--412}.
\newblock
\urldef\tempurl%
\url{http://www.vldb.org/pvldb/vol8/p401-chandramouli.pdf}
\showURL{%
\tempurl}


\bibitem[\protect\citeauthoryear{Chandrasekaran, Cooper, Deshpande, Franklin,
  Hellerstein, Hong, Krishnamurthy, Madden, Reiss, and Shah}{Chandrasekaran
  et~al\mbox{.}}{2003}]%
        {TelegraphCQ}
\bibfield{author}{\bibinfo{person}{Sirish Chandrasekaran},
  \bibinfo{person}{Owen Cooper}, \bibinfo{person}{Amol Deshpande},
  \bibinfo{person}{Michael~J. Franklin}, \bibinfo{person}{Joseph~M.
  Hellerstein}, \bibinfo{person}{Wei Hong}, \bibinfo{person}{Sailesh
  Krishnamurthy}, \bibinfo{person}{Samuel~R. Madden}, \bibinfo{person}{Fred
  Reiss}, {and} \bibinfo{person}{Mehul~A. Shah}.}
  \bibinfo{year}{2003}\natexlab{}.
\newblock \showarticletitle{TelegraphCQ: Continuous Dataflow Processing}. In
  \bibinfo{booktitle}{\emph{Proceedings of the 2003 ACM SIGMOD International
  Conference on Management of Data}} \emph{(\bibinfo{series}{SIGMOD '03})}.
  \bibinfo{publisher}{ACM}, \bibinfo{address}{New York, NY, USA},
  \bibinfo{pages}{668--668}.
\newblock
\showISBNx{1-58113-634-X}
\urldef\tempurl%
\url{https://doi.org/10.1145/872757.872857}
\showDOI{\tempurl}


\bibitem[\protect\citeauthoryear{Dean and Ghemawat}{Dean and Ghemawat}{2008}]%
        {MapReduce}
\bibfield{author}{\bibinfo{person}{Jeffrey Dean} {and} \bibinfo{person}{Sanjay
  Ghemawat}.} \bibinfo{year}{2008}\natexlab{}.
\newblock \showarticletitle{MapReduce: Simplified Data Processing on Large
  Clusters}.
\newblock \bibinfo{journal}{\emph{Commun. ACM}} \bibinfo{volume}{51},
  \bibinfo{number}{1} (\bibinfo{date}{Jan.} \bibinfo{year}{2008}),
  \bibinfo{pages}{107--113}.
\newblock
\showISSN{0001-0782}
\urldef\tempurl%
\url{https://doi.org/10.1145/1327452.1327492}
\showDOI{\tempurl}


\bibitem[\protect\citeauthoryear{Franklin, Tyson, Buckley, Crowley, and
  Maschmeyer}{Franklin et~al\mbox{.}}{2006}]%
        {AutoPipe}
\bibfield{author}{\bibinfo{person}{Mark~A. Franklin}, \bibinfo{person}{Eric~J.
  Tyson}, \bibinfo{person}{James Buckley}, \bibinfo{person}{Patrick Crowley},
  {and} \bibinfo{person}{John Maschmeyer}.} \bibinfo{year}{2006}\natexlab{}.
\newblock \showarticletitle{Auto-pipe and the X Language: A Pipeline Design
  Tool and Description Language}. In \bibinfo{booktitle}{\emph{Proceedings of
  the 20th International Conference on Parallel and Distributed Processing}}
  \emph{(\bibinfo{series}{IPDPS'06})}. \bibinfo{publisher}{IEEE Computer
  Society}, \bibinfo{address}{Washington, DC, USA}, \bibinfo{pages}{117--117}.
\newblock
\showISBNx{1-4244-0054-6}
\urldef\tempurl%
\url{http://dl.acm.org/citation.cfm?id=1898953.1899049}
\showURL{%
\tempurl}


\bibitem[\protect\citeauthoryear{Gordon, Thies, Karczmarek, Lin, Meli, Lamb,
  Leger, Wong, Hoffmann, Maze, and Amarasinghe}{Gordon et~al\mbox{.}}{2002}]%
        {StreamIt}
\bibfield{author}{\bibinfo{person}{Michael~I. Gordon}, \bibinfo{person}{William
  Thies}, \bibinfo{person}{Michal Karczmarek}, \bibinfo{person}{Jasper Lin},
  \bibinfo{person}{Ali~S. Meli}, \bibinfo{person}{Andrew~A. Lamb},
  \bibinfo{person}{Chris Leger}, \bibinfo{person}{Jeremy Wong},
  \bibinfo{person}{Henry Hoffmann}, \bibinfo{person}{David Maze}, {and}
  \bibinfo{person}{Saman~P. Amarasinghe}.} \bibinfo{year}{2002}\natexlab{}.
\newblock \showarticletitle{A stream compiler for communication-exposed
  architectures}. In \bibinfo{booktitle}{\emph{Proceedings of the 10th
  International Conference on Architectural Support for Programming Languages
  and Operating Systems (ASPLOS-X), San Jose, California, USA, October 5-9,
  2002.}} \bibinfo{pages}{291--303}.
\newblock
\urldef\tempurl%
\url{https://doi.org/10.1145/605397.605428}
\showDOI{\tempurl}


\bibitem[\protect\citeauthoryear{Graefe}{Graefe}{1994}]%
        {Volcano}
\bibfield{author}{\bibinfo{person}{G. Graefe}.}
  \bibinfo{year}{1994}\natexlab{}.
\newblock \showarticletitle{Volcano\&\#151 An Extensible and Parallel Query
  Evaluation System}.
\newblock \bibinfo{journal}{\emph{IEEE Trans. on Knowl. and Data Eng.}}
  \bibinfo{volume}{6}, \bibinfo{number}{1} (\bibinfo{date}{Feb.}
  \bibinfo{year}{1994}), \bibinfo{pages}{120--135}.
\newblock
\showISSN{1041-4347}
\urldef\tempurl%
\url{https://doi.org/10.1109/69.273032}
\showDOI{\tempurl}


\bibitem[\protect\citeauthoryear{Hammad, Mokbel, Ali, Aref, Catlin, Elmagarmid,
  Eltabakh, Elfeky, Ghanem, Gwadera, Ilyas, Marzouk, and Xiong}{Hammad
  et~al\mbox{.}}{2004}]%
        {Nile}
\bibfield{author}{\bibinfo{person}{M.~A. Hammad}, \bibinfo{person}{M.~F.
  Mokbel}, \bibinfo{person}{M.~H. Ali}, \bibinfo{person}{W.~G. Aref},
  \bibinfo{person}{A.~C. Catlin}, \bibinfo{person}{A.~K. Elmagarmid},
  \bibinfo{person}{M. Eltabakh}, \bibinfo{person}{M.~G. Elfeky},
  \bibinfo{person}{T.~M. Ghanem}, \bibinfo{person}{R. Gwadera},
  \bibinfo{person}{I.~F. Ilyas}, \bibinfo{person}{M. Marzouk}, {and}
  \bibinfo{person}{X. Xiong}.} \bibinfo{year}{2004}\natexlab{}.
\newblock \showarticletitle{Nile: a query processing engine for data streams}.
  In \bibinfo{booktitle}{\emph{Proceedings. 20th International Conference on
  Data Engineering}}. \bibinfo{pages}{851--}.
\newblock
\showISSN{1063-6382}
\urldef\tempurl%
\url{https://doi.org/10.1109/ICDE.2004.1320080}
\showDOI{\tempurl}


\bibitem[\protect\citeauthoryear{Herlihy and Shavit}{Herlihy and
  Shavit}{2008}]%
        {MauriceHerlihyBook}
\bibfield{author}{\bibinfo{person}{Maurice Herlihy} {and} \bibinfo{person}{Nir
  Shavit}.} \bibinfo{year}{2008}\natexlab{}.
\newblock \bibinfo{booktitle}{\emph{The art of multiprocessor programming}}.
\newblock \bibinfo{publisher}{Morgan Kaufmann}.
\newblock
\showISBNx{978-0-12-370591-4}


\bibitem[\protect\citeauthoryear{Hirzel, Soul{\'{e}}, Schneider, Gedik, and
  Grimm}{Hirzel et~al\mbox{.}}{2014}]%
        {Survey2014}
\bibfield{author}{\bibinfo{person}{Martin Hirzel}, \bibinfo{person}{Robert
  Soul{\'{e}}}, \bibinfo{person}{Scott Schneider}, \bibinfo{person}{Buğra
  Gedik}, {and} \bibinfo{person}{Robert Grimm}.}
  \bibinfo{year}{2014}\natexlab{}.
\newblock \showarticletitle{{A catalog of stream processing optimizations}}.
\newblock \bibinfo{journal}{\emph{Comput. Surveys}} \bibinfo{volume}{46},
  \bibinfo{number}{4} (\bibinfo{date}{mar} \bibinfo{year}{2014}),
  \bibinfo{pages}{1--34}.
\newblock
\showISSN{03600300}
\urldef\tempurl%
\url{https://doi.org/10.1145/2528412}
\showDOI{\tempurl}


\bibitem[\protect\citeauthoryear{Jiang and Chakravarthy}{Jiang and
  Chakravarthy}{2004}]%
        {Jiang2004}
\bibfield{author}{\bibinfo{person}{Qingchun Jiang} {and}
  \bibinfo{person}{Sharma Chakravarthy}.} \bibinfo{year}{2004}\natexlab{}.
\newblock \bibinfo{booktitle}{\emph{Scheduling Strategies for Processing
  Continuous Queries over Streams}}.
\newblock \bibinfo{publisher}{Springer Berlin Heidelberg},
  \bibinfo{address}{Berlin, Heidelberg}, \bibinfo{pages}{16--30}.
\newblock
\showISBNx{978-3-540-27811-5}
\urldef\tempurl%
\url{https://doi.org/10.1007/978-3-540-27811-5_3}
\showDOI{\tempurl}


\bibitem[\protect\citeauthoryear{Kapasi, Dally, Rixner, Owens, and
  Khailany}{Kapasi et~al\mbox{.}}{2002}]%
        {Imagine}
\bibfield{author}{\bibinfo{person}{Ujval Kapasi}, \bibinfo{person}{William~J.
  Dally}, \bibinfo{person}{Scott Rixner}, \bibinfo{person}{John~D. Owens},
  {and} \bibinfo{person}{Brucek Khailany}.} \bibinfo{year}{2002}\natexlab{}.
\newblock \showarticletitle{The {I}magine Stream Processor}. In
  \bibinfo{booktitle}{\emph{Proceedings 2002 IEEE International Conference on
  Computer Design}}. \bibinfo{pages}{282--288}.
\newblock


\bibitem[\protect\citeauthoryear{Kulkarni, Bhagat, Fu, Kedigehalli, Kellogg,
  Mittal, Patel, Ramasamy, and Taneja}{Kulkarni et~al\mbox{.}}{2015}]%
        {Heron}
\bibfield{author}{\bibinfo{person}{Sanjeev Kulkarni}, \bibinfo{person}{Nikunj
  Bhagat}, \bibinfo{person}{Maosong Fu}, \bibinfo{person}{Vikas Kedigehalli},
  \bibinfo{person}{Christopher Kellogg}, \bibinfo{person}{Sailesh Mittal},
  \bibinfo{person}{Jignesh~M. Patel}, \bibinfo{person}{Karthik Ramasamy}, {and}
  \bibinfo{person}{Siddarth Taneja}.} \bibinfo{year}{2015}\natexlab{}.
\newblock \showarticletitle{Twitter Heron: Stream Processing at Scale}. In
  \bibinfo{booktitle}{\emph{Proceedings of the 2015 ACM SIGMOD International
  Conference on Management of Data}} \emph{(\bibinfo{series}{SIGMOD '15})}.
  \bibinfo{publisher}{ACM}, \bibinfo{address}{New York, NY, USA},
  \bibinfo{pages}{239--250}.
\newblock
\showISBNx{978-1-4503-2758-9}
\urldef\tempurl%
\url{https://doi.org/10.1145/2723372.2742788}
\showDOI{\tempurl}


\bibitem[\protect\citeauthoryear{Lee, Leiserson, Schardl, Zhang, and Sukha}{Lee
  et~al\mbox{.}}{2015}]%
        {OnTheFly}
\bibfield{author}{\bibinfo{person}{I-Ting~Angelina Lee},
  \bibinfo{person}{Charles~E. Leiserson}, \bibinfo{person}{Tao~B. Schardl},
  \bibinfo{person}{Zhunping Zhang}, {and} \bibinfo{person}{Jim Sukha}.}
  \bibinfo{year}{2015}\natexlab{}.
\newblock \showarticletitle{On-the-Fly Pipeline Parallelism}.
\newblock \bibinfo{journal}{\emph{ACM Transactions on Parallel Computing}}
  \bibinfo{volume}{2}, \bibinfo{number}{3}, Article \bibinfo{articleno}{17}
  (\bibinfo{date}{Sept.} \bibinfo{year}{2015}), \bibinfo{numpages}{42}~pages.
\newblock
\showISSN{2329-4949}
\urldef\tempurl%
\url{https://doi.org/10.1145/2809808}
\showDOI{\tempurl}


\bibitem[\protect\citeauthoryear{Maier, Li, Tucker, Tufte, and Papadimos}{Maier
  et~al\mbox{.}}{2005}]%
        {NiagraST}
\bibfield{author}{\bibinfo{person}{David Maier}, \bibinfo{person}{Jin Li},
  \bibinfo{person}{Peter Tucker}, \bibinfo{person}{Kristin Tufte}, {and}
  \bibinfo{person}{Vassilis Papadimos}.} \bibinfo{year}{2005}\natexlab{}.
\newblock \showarticletitle{Semantics of Data Streams and Operators}. In
  \bibinfo{booktitle}{\emph{Proceedings of the 10th International Conference on
  Database Theory}} \emph{(\bibinfo{series}{ICDT'05})}.
  \bibinfo{publisher}{Springer-Verlag}, \bibinfo{address}{Berlin, Heidelberg},
  \bibinfo{pages}{37--52}.
\newblock
\showISBNx{3-540-24288-0, 978-3-540-24288-8}
\urldef\tempurl%
\url{https://doi.org/10.1007/978-3-540-30570-5_3}
\showDOI{\tempurl}


\bibitem[\protect\citeauthoryear{Microsoft}{Microsoft}{2016}]%
        {StreamInsight}
\bibfield{author}{\bibinfo{person}{Microsoft}.}
  \bibinfo{year}{2016}\natexlab{}.
\newblock \bibinfo{title}{Microsoft StreamInsight}.
\newblock
  \bibinfo{howpublished}{{\url{https://msdn.microsoft.com/en-us/library/ee362541(v=sql.111).aspx}}}.
    (\bibinfo{year}{2016}).
\newblock
\newblock
\shownote{Accessed: 2017-01-16.}


\bibitem[\protect\citeauthoryear{Murray, McSherry, Isaacs, Isard, Barham, and
  Abadi}{Murray et~al\mbox{.}}{2013}]%
        {Naiad}
\bibfield{author}{\bibinfo{person}{Derek~G. Murray}, \bibinfo{person}{Frank
  McSherry}, \bibinfo{person}{Rebecca Isaacs}, \bibinfo{person}{Michael Isard},
  \bibinfo{person}{Paul Barham}, {and} \bibinfo{person}{Mart\'{\i}n Abadi}.}
  \bibinfo{year}{2013}\natexlab{}.
\newblock \showarticletitle{Naiad: A Timely Dataflow System}. In
  \bibinfo{booktitle}{\emph{Proceedings of the Twenty-Fourth ACM Symposium on
  Operating Systems Principles}} \emph{(\bibinfo{series}{SOSP '13})}.
  \bibinfo{publisher}{ACM}, \bibinfo{address}{New York, NY, USA},
  \bibinfo{pages}{439--455}.
\newblock
\showISBNx{978-1-4503-2388-8}
\urldef\tempurl%
\url{https://doi.org/10.1145/2517349.2522738}
\showDOI{\tempurl}


\bibitem[\protect\citeauthoryear{Murthy, Bhattacharyya, and Lee}{Murthy
  et~al\mbox{.}}{1994}]%
        {SDF2}
\bibfield{author}{\bibinfo{person}{P.~K. Murthy}, \bibinfo{person}{S.~S.
  Bhattacharyya}, {and} \bibinfo{person}{E.~A. Lee}.}
  \bibinfo{year}{1994}\natexlab{}.
\newblock \showarticletitle{Minimizing memory requirements for chain-structured
  synchronous dataflow programs}. In \bibinfo{booktitle}{\emph{Acoustics,
  Speech, and Signal Processing, 1994. ICASSP-94., 1994 IEEE International
  Conference on}}, Vol.~\bibinfo{volume}{ii}. \bibinfo{pages}{II/453--II/456
  vol.2}.
\newblock
\showISSN{1520-6149}
\urldef\tempurl%
\url{https://doi.org/10.1109/ICASSP.1994.389625}
\showDOI{\tempurl}


\bibitem[\protect\citeauthoryear{Neumeyer, Robbins, Nair, and Kesari}{Neumeyer
  et~al\mbox{.}}{2010}]%
        {S4}
\bibfield{author}{\bibinfo{person}{Leonardo Neumeyer}, \bibinfo{person}{Bruce
  Robbins}, \bibinfo{person}{Anish Nair}, {and} \bibinfo{person}{Anand
  Kesari}.} \bibinfo{year}{2010}\natexlab{}.
\newblock \showarticletitle{{S4:} Distributed Stream Computing Platform}. In
  \bibinfo{booktitle}{\emph{{ICDMW} 2010, The 10th {IEEE} International
  Conference on Data Mining Workshops, Sydney, Australia, 13 December 2010}}.
  \bibinfo{pages}{170--177}.
\newblock
\urldef\tempurl%
\url{https://doi.org/10.1109/ICDMW.2010.172}
\showDOI{\tempurl}


\bibitem[\protect\citeauthoryear{Pino, Bhattacharyya, and Lee}{Pino
  et~al\mbox{.}}{1999}]%
        {SDF3}
\bibfield{author}{\bibinfo{person}{Jose~L Pino}, \bibinfo{person}{Shuvra~S.
  Bhattacharyya}, {and} \bibinfo{person}{Edward~A. Lee}.}
  \bibinfo{year}{1999}\natexlab{}.
\newblock \bibinfo{booktitle}{\emph{A Hierarchical Multiprocessor Scheduling
  Framework for Synchronous Dataflow Graphs}}.
\newblock \bibinfo{type}{{T}echnical {R}eport}. \bibinfo{address}{Berkeley, CA,
  USA}.
\newblock


\bibitem[\protect\citeauthoryear{Schneider and Wu}{Schneider and Wu}{2017}]%
        {IBMStreams}
\bibfield{author}{\bibinfo{person}{Scott Schneider} {and}
  \bibinfo{person}{Kun-Lung Wu}.} \bibinfo{year}{2017}\natexlab{}.
\newblock \showarticletitle{Low-synchronization, Mostly Lock-free, Elastic
  Scheduling for Streaming Runtimes}. In \bibinfo{booktitle}{\emph{Proceedings
  of the 38th ACM SIGPLAN Conference on Programming Language Design and
  Implementation}} \emph{(\bibinfo{series}{PLDI 2017})}.
  \bibinfo{publisher}{ACM}, \bibinfo{address}{New York, NY, USA},
  \bibinfo{pages}{648--661}.
\newblock
\showISBNx{978-1-4503-4988-8}
\urldef\tempurl%
\url{https://doi.org/10.1145/3062341.3062366}
\showDOI{\tempurl}


\bibitem[\protect\citeauthoryear{Shah, Hellerstein, Chandrasekaran, and
  Franklin}{Shah et~al\mbox{.}}{2003}]%
        {Flux2003:ICDE}
\bibfield{author}{\bibinfo{person}{Mehul~A. Shah}, \bibinfo{person}{Joseph~M.
  Hellerstein}, \bibinfo{person}{Sirish Chandrasekaran}, {and}
  \bibinfo{person}{Michael~J. Franklin}.} \bibinfo{year}{2003}\natexlab{}.
\newblock \showarticletitle{Flux: An Adaptive Partitioning Operator for
  Continuous Query Systems}. In \bibinfo{booktitle}{\emph{Proceedings of the
  19th International Conference on Data Engineering, March 5-8, 2003,
  Bangalore, India}}. \bibinfo{pages}{25--36}.
\newblock
\urldef\tempurl%
\url{https://doi.org/10.1109/ICDE.2003.1260779}
\showDOI{\tempurl}


\bibitem[\protect\citeauthoryear{Thies}{Thies}{2009}]%
        {Thies2009}
\bibfield{author}{\bibinfo{person}{William Thies}.}
  \bibinfo{year}{2009}\natexlab{}.
\newblock \emph{\bibinfo{title}{{[StreamIt] Language and Compiler Support for
  Stream Programs}}}.
\newblock \bibinfo{thesistype}{Ph.D. Dissertation}.
\newblock


\bibitem[\protect\citeauthoryear{Zaharia, Chowdhury, Franklin, Shenker, and
  Stoica}{Zaharia et~al\mbox{.}}{2010}]%
        {Spark}
\bibfield{author}{\bibinfo{person}{Matei Zaharia}, \bibinfo{person}{Mosharaf
  Chowdhury}, \bibinfo{person}{Michael~J. Franklin}, \bibinfo{person}{Scott
  Shenker}, {and} \bibinfo{person}{Ion Stoica}.}
  \bibinfo{year}{2010}\natexlab{}.
\newblock \showarticletitle{Spark: Cluster Computing with Working Sets}. In
  \bibinfo{booktitle}{\emph{Proceedings of the 2Nd USENIX Conference on Hot
  Topics in Cloud Computing}} \emph{(\bibinfo{series}{HotCloud'10})}.
  \bibinfo{publisher}{USENIX Association}, \bibinfo{address}{Berkeley, CA,
  USA}, \bibinfo{pages}{10--10}.
\newblock
\urldef\tempurl%
\url{http://dl.acm.org/citation.cfm?id=1863103.1863113}
\showURL{%
\tempurl}


\bibitem[\protect\citeauthoryear{Zaharia, Das, Li, Shenker, and Stoica}{Zaharia
  et~al\mbox{.}}{2012}]%
        {SparkStreaming}
\bibfield{author}{\bibinfo{person}{Matei Zaharia}, \bibinfo{person}{Tathagata
  Das}, \bibinfo{person}{Haoyuan Li}, \bibinfo{person}{Scott Shenker}, {and}
  \bibinfo{person}{Ion Stoica}.} \bibinfo{year}{2012}\natexlab{}.
\newblock \showarticletitle{Discretized Streams: An Efficient and
  Fault-tolerant Model for Stream Processing on Large Clusters}. In
  \bibinfo{booktitle}{\emph{Proceedings of the 4th USENIX Conference on Hot
  Topics in Cloud Ccomputing}} \emph{(\bibinfo{series}{HotCloud'12})}.
  \bibinfo{publisher}{USENIX Association}, \bibinfo{address}{Berkeley, CA,
  USA}, \bibinfo{pages}{10--10}.
\newblock
\urldef\tempurl%
\url{http://dl.acm.org/citation.cfm?id=2342763.2342773}
\showURL{%
\tempurl}


\end{thebibliography}
\end{document}